\documentclass[aps,pra,preprint,showkeys]{revtex4-2}

\usepackage[T1]{fontenc}
\usepackage{amsmath,amssymb,amsthm}
\usepackage{mathrsfs}
\usepackage{enumitem}
\setlist[enumerate,1]{label={(\roman*)}}
\usepackage{hyperref}
\hypersetup{colorlinks=true,linkcolor=blue,citecolor=red,urlcolor=blue}
\usepackage{graphicx}
\usepackage{tikz}
\usetikzlibrary{positioning,arrows.meta,calc}

\theoremstyle{definition}
\newtheorem{definition}{Definition}[section]
\newtheorem{remark}[definition]{Remark}
\theoremstyle{plain}
\newtheorem{theorem}[definition]{Theorem}
\newtheorem{proposition}[definition]{Proposition}
\newtheorem{lemma}[definition]{Lemma}

\newcommand{\bra}{\langle}
\newcommand{\ket}{\rangle}

\newcommand{\bs}[1]{\boldsymbol{#1}}
\newcommand{\mcal}[1]{\mathcal{#1}}
\newcommand{\mrm}[1]{\mathrm{#1}}
\newcommand{\mbb}[1]{\mathbb{#1}}
\newcommand{\op}[1]{\operatorname{#1}}
\newcommand{\im}{\mathrm{i}}
\newcommand{\ran}{\operatorname{ran}}
\newcommand{\Tr}{\operatorname{Tr}}

\newcommand{\numberthis}{\addtocounter{equation}{1}\tag{\theequation}} 

\newcommand{\MmR}{\mathcal{M}_m(\mathbb{R})}
\newcommand{\Cquantum}{\mcal{C}_{\mathrm{qm}}^{AB}}
\newcommand{\Csep}{\mcal{C}_{\mathrm{sep}}^{AB}}
\newcommand{\Cmax}{\mcal{C}_{\mathrm{max}}^{AB}}
\newcommand{\Cmodel}{\mcal{C}_{\mathrm{model}}^{AB}}
\newcommand{\D}{\mathcal{D}}

\newcommand{\CtwoCtwo}{\mathbb{C}^2 \otimes \mathbb{C}^2}
\newcommand{\phiquantum}{\phi_{\mathrm{qm}}^{AB}}
\newcommand{\phisep}{\phi_{\mathrm{sep}}^{AB}}
\newcommand{\gammaquantum}{\gamma_{\mathrm{qm}}^{AB}}
\newcommand{\gammasep}{\gamma_{\mathrm{sep}}^{AB}}
\newcommand{\gammamax}{\gamma_{\mathrm{max}}^{AB}}
\newcommand{\phimax}{\phi_{\mathrm{max}}^{AB}}

\newcommand{\Ssep}{\mcal{S}_{\mrm{sep}}}
\newcommand{\Squantum}{\mcal{S}_{\mrm{qm}}}
\newcommand{\Smax}{\mcal{S}_{\mrm{max}}}
\newcommand{\norm}[1]{\left\|#1\right\|}
\newcommand{\opnorm}[1]{\norm{#1}_{\infty}}
\newcommand{\trnorm}[1]{\norm{#1}_{1}}
\newcommand{\rank}{\operatorname{rank}}

\newcommand{\spi}{-1}

\newcommand{\HS}{2}
\newcommand{\ep}[1]{\mathrm{e}^{#1}}

\begin{document}

\title{Dualistic operational characterization of device-dependent correlation sets via convex analysis in the $(2,m,2)$ Bell scenario}

\author{Ryosuke Nogami}
\email{ryosuke.nogami1997@gmail.com}
\affiliation{Graduate School of Informatics, Nagoya University, Furo-cho, Chikusa-Ku, Nagoya, 466-8601, Japan}

\author{Jaeha Lee}
\email{lee@iis.u-tokyo.ac.jp}
\affiliation{Institute of Industrial Science, the University of Tokyo, 5-1-5 Kashiwanoha, Kashiwa, Chiba, 277-8574, Japan}

\begin{abstract}

We analyze device-dependent correlation sets generated by fixed local dichotomic measurements for two-qubit systems in the $(2,m,2)$ Bell scenario.
We consider three fundamental state spaces for the composite system: the separable state space, the standard quantum state space, and the maximal tensor-product state space, which contains beyond-quantum states compatible with local quantum measurements.
We formulate the corresponding correlation sets for general fixed dichotomic measurements and, in the traceless case, derive particularly simple explicit formulae for their support and gauge functions.
These functions furnish dual operational characterizations of the three correlation sets: the support functions give optimal witnesses for entanglement and beyond-quantum states, whereas the gauge functions quantify the robustness of these detections against depolarizing noise.
We further derive convex-hull representations that elucidate the extremal structures of the correlation sets and the physical states realizing them, showing in particular that extremal quantum correlations are realized by maximally entangled states.
The fundamental limits of these dual operational tasks are governed solely by the smaller of the numbers of linearly independent measurement directions available to Alice and Bob.
When both parties have three linearly independent measurement directions, our entanglement criterion detects Werner states up to the optimal PPT threshold $p_{\mathrm{crit}}=2/3$.
For beyond-quantum-state detection, a nontrivial separation from the quantum set occurs only under the same measurement condition; in that case, the same optimal noise threshold is attained for an extremal state in the maximal tensor-product state space.

\end{abstract}


\maketitle

\section{Introduction}
Quantum entanglement is a defining feature of composite quantum systems and a central resource in quantum information science~\cite{horodecki2009quantum}.
Determining whether a given state is entangled and quantifying its entanglement are fundamental problems in the field~\cite{guhne2009entanglement}.
Among the practical tools developed for these tasks, entanglement witnesses~\cite{terhal2000bell,lewenstein2000optimization} detect entanglement via operator expectation values, while robustness measures~\cite{vidal1999robustness,steiner2003generalized} quantify the tolerance of entanglement to noise.

Bell~\cite{bell1964einstein} demonstrated that entangled states can produce correlations that violate Bell inequalities, which constrain all local realistic theories. 
Conversely, a violation of a Bell inequality certifies the presence 
of entanglement~\cite{terhal2000bell}.
The Clauser--Horne--Shimony--Holt (CHSH) inequality~\cite{clauser1969proposed,clauser1978bell} provides a simple test of this type for bipartite systems with two dichotomic measurements per party, and the maximum quantum violation is given by the Tsirel'son bound~\cite{cirel1980quantum}.
In the CHSH scenario, the set of quantum-realizable correlations is completely characterized by the Tsirel'son--Landau--Masanes (TLM) criterion~\cite{tsirelson1987quantum,tsirelson1993some,landau1988empirical,masanes2003necessarysufficient}, and its geometry as a convex body has been studied~\cite{Goh2018geometry,Thinh2019geometric,Le2023quantumcorrelations}. For general Bell scenarios, the Navascu\'es--Pironio--Ac\'in (NPA) hierarchy~\cite{navascues2007bounding,navascues2008convergent} provides a systematic semidefinite programming approach to outer-approximate the quantum set.

In the standard approach to Bell-inequality studies, the analysis is \emph{device-independent}~\cite{scarani2012device}: the constraints hold for all possible choices of measurement settings.
While this universality is essential for foundational tests~\cite{brunner2014bell,Supic2020selftestingof} and device-independent protocols~\cite{Acin2007device-independent,Pironio_2009,Vazirani2019fully}, 
it discards the structure imposed by a specific choice of observables.
When the measurement settings are fixed, the set of realizable correlations satisfies stronger constraints, and therefore carries more detailed information~\cite{Selby2023correlations}; we refer to this as the \emph{device-dependent} setting.
An inequality that characterizes device-dependent quantum-realizable correlations in the CHSH scenario was derived in Ref.~\cite{nogami2025necessary}, and Le \textit{et al.}~\cite{Le2023quantumcorrelations} provided a convex-geometric description of the device-dependent quantum correlation set in the same $(2,2,2)$ setting.
However, explicit characterizations of device-dependent correlation sets for general $(2,m,2)$ scenarios---and, in particular, for the separable and beyond-quantum sets---have remained largely unexplored.

The distinction between device-dependent and device-independent analyses is also central to the detection of beyond-quantum correlations.
In the framework of generalized probabilistic theories (GPTs)~\cite{barrett2007information,Janotta_2014,PLAVALA2023general}, 
one natural class of beyond-quantum models is provided by \emph{entanglement structures (ESs)}~\cite{Janotta_2014,arai2019perfect,arai2023pseudo}: composite systems in which each local subsystem retains the full quantum-theoretical state space, but the joint state space is no longer required to be the set of density operators on the tensor product Hilbert space.
The extremal cases are the \emph{minimal} tensor product, which 
recovers the separable states, and the \emph{maximal} tensor product, whose 
state space is the dual of the separable cone.
Arai, Yu, and Hayashi~\cite{arai2024detecting} demonstrated that a device-dependent protocol in the $(2,3,2)$ scenario (two parties, three dichotomic measurements each) can distinguish beyond-quantum states in ESs from quantum states, whereas no device-independent protocol can~\cite{Barnum2010local,acin2010unified}.
In the $(2,2,2)$ scenario, i.e., the CHSH scenario, however, correlations arising from any beyond-quantum ES model cannot be distinguished from quantum-realizable ones~\cite{Banik2013degree,Stevens2014steering,barnum2013ensemble}.
This motivates a systematic study of general Bell scenarios with more than two measurement choices per party.

Since the device-dependent and device-independent correlation sets are convex and compact, tools from convex analysis provide a natural framework for their study.
Convex-geometric methods have proved fruitful in analyzing both the device-independent quantum correlation set~\cite{Goh2018geometry,Thinh2019geometric,Le2023quantumcorrelations} and quantum resource theories more broadly~\cite{regula2018convex}.
In operator space theory, the degree of violation of a Bell 
inequality is quantified by tensor norms of the corresponding Bell operator~\cite{junge2010operator}, establishing a deep connection between the geometry of correlation sets and the algebraic structure of operator spaces.
Convex-analytic tools apply equally to the device-dependent setting: the support and gauge functions form a dual pair that together encode the complete geometry of a convex body, and their explicit expressions capture the structure imposed by the choice of observables.

In this paper, we provide a complete dual characterization of the three device-dependent correlation sets---produced by separable states, quantum states, and states in the maximal entanglement structure---in the $(2,m,2)$ scenario for two-qubit systems $\CtwoCtwo$: we derive explicit formulas for the support functions and gauge functions of all three sets (Sections~\ref{sec:support_functions} and~\ref{sec:gauge_functions}). 
The support and gauge functions of each set take the form of a pair of dual matrix norms---standard norms for the separable and maximal correlation sets, and an asymmetric dual pair for the quantum correlation set---whose arguments encode the dependence on the measurement settings.
We then apply these functions to the detection of entanglement and beyond-quantum states (Sections~\ref{sec:detection_of_entanglement} and~\ref{sec:detection_of_BQS}).
The gauge functions of the separable and quantum correlation sets serve as robustness measures for entanglement and beyond-quantum detection, respectively, and the fundamental limits of these detection sensitivities are expressed as the suprema of the ratios of the corresponding support functions.

The remainder of this paper is organized as follows. Section~\ref{sec:preliminaries} introduces the notation, defines the device-dependent correlation sets for the three models of entanglement structures, and reviews the relevant concepts from convex analysis. 
Sections~\ref{sec:support_functions} and~\ref{sec:gauge_functions} present the support function and gauge function theorems, respectively.
Section~\ref{sec:detection_of_entanglement} applies these results to entanglement detection, and Section~\ref{sec:detection_of_BQS} extends the analysis to beyond-quantum detection. 
Section~\ref{sec:discussion} discusses the structural insights that emerge from the dual characterization and directions for future research.
Section~\ref{sec:conclusion} provides a summary of the paper.

\section{Preliminaries}\label{sec:preliminaries}
\subsection{Notation}
Throughout this paper, we adopt the following notational conventions:
\begin{itemize}
    \item Vectors are denoted by boldface lowercase letters such as $\bs{a}$.
    \item Matrices are denoted by uppercase letters such as $A$, with components $A = (a_{ij})$ where $a_{ij}$ denotes the $(i,j)$-th entry.
    \item The set of $m \times n$ matrices over a field $\mathbb{K}$ is denoted by $\mathcal{M}_{m,n}(\mathbb{K})$; the square case is abbreviated as $\mathcal{M}_n(\mathbb{K}) := \mathcal{M}_{n,n}(\mathbb{K})$.
    \item Quantum states in a Hilbert space are represented using Dirac bra-ket notation, e.g., $|\psi\ket$ for kets and $\bra\psi|$ for bras.
    \item Operators acting on a Hilbert space are denoted with a hat, e.g., $\hat{X}$.
    \item The set of bounded operators acting on a Hilbert space $\mcal{H}$ is denoted by $\mcal{B}(\mcal{H})$. The real subspace of self-adjoint operators is denoted by $\mcal{B}^{\mrm{S}}(\mcal{H})$, and the cone of positive semidefinite operators by $\mcal{B}^{+}(\mcal{H}) := \{\hat{X} \in \mcal{B}^{\mrm{S}}(\mcal{H}) \,:\, \hat{X} \geq 0\}$.
\end{itemize}
\subsection{The $(2,m,2)$ Scenario}
We consider the $(2,m,2)$ scenario, where two parties (Alice and Bob) each perform $m$ dichotomic measurements with outcomes in $\{\pm 1\}$.

We seek realizability conditions for correlations in the Hilbert space $\CtwoCtwo$ under given measurement settings. The measurement observables for Alice and Bob are represented by self-adjoint operators $\hat{A}_i$ and $\hat{B}_j$ on $\CtwoCtwo$ with spectra $\{\pm 1 \}$, where $i,j \in \{1,\ldots,m\}$.
The observables are characterized by unit vectors $\bs{a}_i$ and $\bs{b}_j$ in $\mathbb{R}^3$ such that
\begin{align}
    \hat{A}_i &= \bs{a}_i \cdot \hat{\bs{\sigma}} := \sum_{k=1}^3 a_{ik} \hat{\sigma}_k, \\
    \hat{B}_j &= \bs{b}_j \cdot \hat{\bs{\sigma}} := \sum_{k=1}^3 b_{jk} \hat{\sigma}_k.
\end{align}
Thus, the measurement settings are fully specified by the matrices 
\begin{align}
    A &:= (a_{ij}) = \begin{bmatrix}
        \bs{a}_1^\top \\ \vdots \\ \bs{a}_m^\top
    \end{bmatrix} \in \mcal{M}_{m,3}(\mbb{R}), \\
    B &:= (b_{ij}) = \begin{bmatrix}
        \bs{b}_1^\top \\ \vdots \\ \bs{b}_m^\top
    \end{bmatrix} \in \mcal{M}_{m,3}(\mbb{R}).
\end{align}

\subsection{Entanglement Structures}\label{subsec:ES}
In the framework of GPTs, a composite system whose local subsystems are described by standard quantum theory need not possess the standard quantum state space as its joint state space.
An \emph{entanglement structure} (ES)~\cite{Janotta_2014,arai2019perfect,arai2023pseudo,arai2024detecting} specifies a model of such a composite system by choosing a cone of valid states subject to consistency with local quantum measurements.
We restrict the exposition to the two-qubit system $\CtwoCtwo$; for the general multipartite definition, see Refs.~\cite{Janotta_2014,arai2024detecting,arai2024detecting_supplemental}.

We equip the real vector space $\mcal{B}^{\mrm{S}}(\CtwoCtwo)$ with the Hilbert--Schmidt inner product $\bra \hat{X}, \hat{Y} \ket_{\HS} := \Tr[\hat{X}\hat{Y}]$.
The \emph{separable cone} is defined as
\begin{equation}\label{eq:sep_cone}
    \op{SEP} := \op{Conv}\!\bigl\{\hat{X}_A \otimes \hat{X}_B \,:\, \hat{X}_A,\; \hat{X}_B \in \mcal{B}^{+}(\mbb{C}^2)\bigr\},
\end{equation}
where $\op{Conv}(\cdot)$ denotes the convex hull.
This cone is the minimal tensor product of the local positive cones.
Its dual cone with respect to the Hilbert--Schmidt inner product is
\begin{equation}\label{eq:sep_dual}
    \op{SEP}^{*} := \bigl\{\hat{Y} \in \mcal{B}^{\mrm{S}}(\CtwoCtwo) \,:\, \forall \hat{X} \in \op{SEP}, \; \Tr[\hat{Y}\hat{X}] \geq 0\bigr\}.
\end{equation}

\begin{definition}[Entanglement structure~{\cite{Janotta_2014,arai2024detecting,arai2024detecting_supplemental}}]\label{def:ES}
An entanglement structure for the two-qubit system $\CtwoCtwo$ is a proper cone $\mcal{K}\subset \mcal{B}^{\mrm{S}}(\CtwoCtwo)$ satisfying
\begin{equation}\label{eq:ES_sandwich}
    \op{SEP} \subset \mcal{K}\subset \op{SEP}^{*}.
\end{equation}
The state space associated with $\mcal{C}$ is $\mcal{S}(\mcal{C}) := \{\hat{\rho} \in \mcal{K}: \Tr[\hat{\rho}] = 1\}$.
\end{definition}

The left inclusion in~\eqref{eq:ES_sandwich} guarantees that every separable state belongs to $\mcal{S}(\mcal{K})$, while the right inclusion ensures that every state in $\mcal{S}(\mcal{K})$ yields nonnegative probabilities for all measurements composed of local POVMs~\cite{arai2024detecting,arai2024detecting_supplemental}.
Three canonical choices of $\mcal{C}$ define the models studied in this paper:
\begin{enumerate}
    \item \textit{Separable states (minimal tensor product).}
    Setting $\mcal{K}= \op{SEP}$ yields the state space of separable states:
    \begin{align*}\label{eq:S_sep}
        \Ssep :=& \mcal{S}(\op{SEP}) \\
        =& \op{Conv}\!\bigl\{\hat{\rho}_A \otimes \hat{\rho}_B \,:\, \hat{\rho}_A,\;\hat{\rho}_B \in \D(\mbb{C}^2)\bigr\},
        \numberthis
    \end{align*}
    where $\D(\mcal{H}) := \{\hat{\rho} \in \mcal{B}^{+}(\mcal{H}) : \Tr[\hat{\rho}] = 1\}$ denotes the set of density operators on $\mcal{H}$.
    This is the smallest admissible state space.
    \item \textit{Quantum states (standard entanglement structure).}
    Setting $\mcal{K}= \mcal{B}^{+}(\CtwoCtwo)$, i.e., the cone of all positive semidefinite operators, recovers the standard quantum state space:
    \begin{align*}\label{eq:S_qm}
        \Squantum :=& \D(\CtwoCtwo) \\
        =& \bigl\{\hat{\rho} \in \mcal{B}^{+}(\CtwoCtwo) \,:\, \Tr[\hat{\rho}] = 1\bigr\}.
        \numberthis
    \end{align*}
    Since $\op{SEP} \subset \mcal{B}^{+}(\CtwoCtwo) \subset \op{SEP}^{*}$, this cone defines a valid ES, called the \emph{standard entanglement structure} (SES)~\cite{arai2024detecting}.
    \item \textit{Maximal state space (maximal tensor product).}
    Setting $\mcal{K}= \op{SEP}^{*}$ yields the largest admissible state space:
    \begin{align*}\label{eq:S_max}
        \Smax :=& \mcal{S}(\op{SEP}^{*}) \\
        =& \bigl\{\hat{\rho} \in \mcal{B}^{\mrm{S}}(\CtwoCtwo) \,:\, \Tr[\hat{\rho}] = 1,\;
         \forall \hat{M} \in \op{SEP}, \, \Tr[\hat{\rho}\,\hat{M}] \geq 0\bigr\}.\numberthis
    \end{align*}
    A state $\hat{\rho} \in \mcal{S}_{\mrm{max}} \setminus \mcal{S}_{\mrm{qm}}$ is called a \emph{beyond-quantum state}~\cite{arai2024detecting}.
\end{enumerate}
These three state spaces satisfy the strict inclusions
\begin{align}\label{eq:state_space_inclusions}
    \mcal{S}_{\mrm{sep}} \subsetneq \mcal{S}_{\mrm{qm}} \subsetneq \mcal{S}_{\mrm{max}}.
\end{align}
Every ES for $\CtwoCtwo$ has a state space $\mcal{S}(\mcal{K})$ lying between $\mcal{S}_{\mrm{sep}}$ and $\mcal{S}_{\mrm{max}}$, and the standard quantum state space $\mcal{S}_{\mrm{qm}}$ is one particular instance within this range.

\subsection{Device-Dependent Correlation Sets}
For fixed measurement settings $A$ and $B$, we define a linear mapping that sends a state to the matrix whose $(i,j)$-th entry is given by the correlation between $\hat{A}_i$ and $\hat{B}_j$:
\begin{equation}
    \Phi_{AB} : \mcal{B}(\CtwoCtwo) \to \MmR, \; 
    \hat{\rho} \mapsto C=(c_{ij}),
\end{equation}
where
\begin{equation}
    c_{ij} = \Tr\left[\hat{\rho} \left(\hat{A}_i \otimes \hat{B}_j\right)\right] \; \text{for all} \; i,j \in \{1,\ldots,m\}.
\end{equation}

We consider the correlation sets corresponding to the three state spaces defined above:
\begin{enumerate}
    \item \textit{Separable correlation set}. 
    The set of correlation matrices realized by separable states:
    \begin{equation}
        \Csep := \{\Phi_{AB}(\hat{\rho}) \in \MmR \,:\, \hat{\rho} \in \Ssep\}.
    \end{equation}
    \item \textit{Quantum correlation set}.
    The set of correlation matrices realized by quantum states:
    \begin{equation}
        \Cquantum := \{\Phi_{AB}(\hat{\rho}) \in \MmR \,:\, \hat{\rho} \in \Squantum\}.
    \end{equation}
    \item \textit{Maximal correlation set}.
    The set of correlation matrices realized by states in the maximal entanglement structure:
    \begin{equation}
        \Cmax := \{\Phi_{AB}(\hat{\rho}) \in \MmR \,:\, \hat{\rho} \in \Smax\}.
    \end{equation}
\end{enumerate}
We have the inclusions
\begin{equation}
    \Csep \subset \Cquantum \subset \Cmax,
\end{equation}
but neither inclusion is necessarily strict, since the mapping $\Phi_{AB}$ need not be injective.

\subsection{Convex Analysis: Support and Gauge Functions}
We introduce two convex-analytic tools: support and gauge functions.

Let $(V, \bra \cdot, \cdot \ket_V)$ be a real inner product space, and let $K \subset V$.

\begin{definition}[Support Function]
    The \emph{support function} of the set $K$ is defined by
    \begin{equation}
        \phi_K (Z) := \sup_{C \in K} \bra Z, C \ket_V, \quad \forall Z \in V.
    \end{equation}
\end{definition}
If $K$ is bounded, $\phi_K$ is Lipschitz continuous: for any $Z, Z' \in V$,
\begin{equation}\label{eq:support_function_lipschitz}
    |\phi_K(Z) - \phi_K(Z')| \leq \sup_{C \in K} \|C\|_V \, \|Z - Z'\|_V,
\end{equation}
where $\|\cdot\|_V$ denotes the norm induced by $\bra \cdot, \cdot \ket_V$.

If $K$ is closed and convex, the support function $\phi_K$ completely characterizes $K$:
\begin{equation}
    K = \bigcap_{Z \in V} \{C \in V \,:\, \bra Z, C \ket_V \leq \phi_K(Z)\}.
\end{equation}
\begin{definition}[Gauge Function]
    The \emph{gauge function} of the set $K$ is defined by
    \begin{equation}
        \gamma_K (C) := \inf \{t > 0 \,:\, C/t \in K\}, \quad \forall C \in V.
    \end{equation}
\end{definition}
If $K$ is a closed convex set containing the origin, then $K$ is the unit ball of its gauge function:
\begin{equation}
    K = \{C \in V \,:\, \gamma_K(C) \leq 1\}.
\end{equation}

Under the same assumptions, the support and gauge functions are dual to each other in the following sense:
\begin{align}
    \phi_K(Z) &= \sup_{C \neq 0} \frac{\bra Z, C \ket_V}{\gamma_K (C)}
    = \sup_{C \,:\, \gamma_K (C) \leq 1} \bra Z, C \ket_V, \\
    \gamma_K(C) &= \sup_{Z \neq 0} \frac{\bra Z, C \ket_V}{\phi_K(Z)}
    = \sup_{Z \,:\, \phi_K(Z) \leq 1} \bra Z, C\ket_V.
    \label{eq:gauge_function_as_a_dual_of_support_function}
\end{align}

\section{Support functions}\label{sec:support_functions}

We equip the real vector space $\MmR$ with the Hilbert--Schmidt inner product $\bra X, Y \ket_{\HS} := \Tr[X^\top Y]$.

In this section, we derive explicit formulae for the support functions
\begin{equation}
    \phi^{AB}_{\text{model}}(Z)  := \phi_{\Cmodel}(Z)
    = \sup_{C \in \Cmodel} \Tr[Z^\top C] , \quad
    \text{model} = \text{sep},\,\text{qm}, \,\text{max}.
\end{equation}


In what follows, we show that all three support functions are determined by the matrix $A^\top Z B \in \mcal{M}_{3}(\mbb{R})$. Let $s_1 \geq s_2 \geq s_3 \geq 0$ denote the singular values of $A^\top Z B$. 

\subsection{Support function of $\Csep$}

\begin{theorem}[{Support function of $\Csep$}]\label{thm:support_function_of_Csep}
    The support function $\phisep$ of $\Csep$ is given by
    \begin{equation}
        \phisep (Z) = \opnorm{A^\top Z B} = s_1,
        \label{eq:support_function_of_Csep_as_a_operator_norm}
    \end{equation}
    where $\| \cdot \|_{\infty}$ denotes the operator norm. 

    When $m=2$, this specializes to
    \begin{equation}
        \phisep (Z) = \sqrt{
            \frac{\Tr[G_{\alpha}\, Z\, G_{\beta}\, Z^\top]
            + \left| \Tr[G_{\alpha}\, Z\, H_{\beta}\, Z^\top]\right|
            }{2}
        },
        \label{eq:support_function_of_Csep_matrix_form}
    \end{equation}
    or equivalently,
    \begin{equation}
        \phisep (Z) = \sqrt{
            \frac{\bs{z}^\top K_{\alpha,\beta} \bs{z}
            + \left| \bs{z}^\top J_{\alpha,\beta} \bs{z}\right|
            }{2}
        },
        \label{eq:support_function_of_Csep_as_a_quadratic_form}
    \end{equation}
    where
    \begin{gather}
        K_{\alpha,\beta} := G_{\beta} \otimes G_{\alpha},\quad
        J_{\alpha,\beta} := H_{\beta} \otimes G_{\alpha},\\
         G_{\theta} := \begin{pmatrix}
             1 & \cos \theta \\
             \cos \theta & 1
         \end{pmatrix},
         \quad H_{\theta} := \begin{pmatrix}
            \ep{-\im \theta} & 1 \\
            1 & \ep{\im \theta}
         \end{pmatrix}
         ,\\
         \alpha := \arccos \bs{a}_1 \cdot \bs{a}_2, \quad \beta := \arccos \bs{b}_1 \cdot \bs{b}_2,\\
         \bs{z}:=(z_{11},z_{21},z_{12},z_{22})^\top.
    \end{gather}
\end{theorem}
\begin{proof}
    See Appendix~\ref{app:proof_support_Csep}.
\end{proof}

\subsection{Support function of $\Cquantum$}

\begin{theorem}[{Support function of $\Cquantum$}]\label{thm:support_function_of_Cquantum}
    The support function $\phiquantum$ of $\Cquantum$ is given by
    \begin{equation}
        \phiquantum (Z) =
        s_1 + s_2 - \epsilon \, s_3,
        \label{eq:support_function_of_Cquantum_as_a_nuclear_norm}
    \end{equation}
    where $\epsilon := \op{sgn} (\det(A^\top ZB))$.

    When $m=2$, this specializes to
    \begin{equation}
        \phiquantum (Z) = \sqrt{\Tr[G_{\alpha}\, Z\, G_{\beta}\, Z^\top] + 2|\det Z| \sin\alpha \sin\beta},
        \label{eq:support_function_of_Cquantum_matrix_form}
    \end{equation}
    or equivalently,
    \begin{equation}
        \phiquantum (Z) = \sqrt{\bs{z}^\top K_{\alpha,\beta} \bs{z} + 2|\det Z| \sin\alpha \sin\beta}.
        \label{eq:support_function_of_Cquantum_as_a_quadratic_form}
    \end{equation}

\end{theorem}
\begin{proof}
    See Appendix~\ref{app:proof_support_Cquantum}.
\end{proof}

\subsection{Support function of $\Cmax$}
\begin{theorem}[{Support function of $\Cmax$}]\label{thm:support_function_of_Cmax}
    The support function $\phimax$ of $\Cmax$ is given by
    \begin{equation}
        \phimax (Z) = \trnorm{A^\top Z B}=s_1+s_2+s_3,
    \end{equation}
    where $\|\cdot\|_1$ denotes the trace norm.
\end{theorem}
\begin{proof}
    See Appendix~\ref{app:proof_support_Cmax}.
\end{proof}
Note that $\phiquantum(Z) = \|A^\top Z B\|_1 = \phimax(Z)$ holds when $m=2$, since the $3 \times 3$ matrix $A^\top Z B$ has rank at most~$2$, so $s_3 = 0$.
Therefore, $\Cquantum = \Cmax$ when $m=2$.

\section{Gauge functions}\label{sec:gauge_functions}
In this section, we derive explicit formulae for the gauge functions
\begin{equation}
    \gamma^{AB}_{\text{model}}(C)  := \gamma_{\Cmodel}(C)
    = \inf \{t > 0 \,:\, C/t \in \Cmodel\}, \quad
    \text{model} = \text{sep},\,\text{qm}, \,\text{max}.
\end{equation}


In what follows, we show that all three gauge functions are determined by the matrix $A^{+}C(B^\top)^{+} \in \mcal{M}_{3}(\mbb{R})$, where $X^{+}$ denotes the Moore--Penrose pseudoinverse of a matrix $X$. Let $s'_1 \geq s'_2 \geq s'_3 \geq 0$ denote the singular values of $A^{+}C(B^\top)^{+}$.

\subsection{Gauge function of $\Csep$}
\begin{theorem}[{Gauge function of $\Csep$}]\label{thm:gauge_function_of_Csep}
    The gauge function $\gammasep$ of $\Csep$ is given by
    \begin{equation}
        \gammasep (C) =
        \begin{cases}
            \trnorm{A^{+} C\, (B^\top)^{+}}=s'_1+s'_2+s'_3 & \text{if~} \ran C \subset \ran A,~ \ran C^\top \subset \ran B, \\
            +\infty & \text{otherwise}.
        \end{cases}
    \end{equation}

    When $m=2$ and $\sin\alpha\sin\beta \neq 0$, this specializes to
    \begin{equation}
        \gammasep (C) = \sqrt{\Tr[G_{\alpha}^{-1}\, C\, G_{\beta}^{-1}\, C^\top] + \frac{2|\det C|}{\sin\alpha\sin\beta}},
        \label{eq:gauge_function_of_Csep_matrix_form}
    \end{equation}
    or equivalently,
    \begin{equation}
        \gammasep (C) = \sqrt{\bs{c}^\top K_{\alpha,\beta}^{-1} \bs{c} + \frac{2|\det C|}{\sin\alpha\sin\beta}},
    \end{equation}
    where
    \begin{equation}
        \bs{c} := (c_{11},c_{21},c_{12},c_{22})^\top.
    \end{equation}
\end{theorem}
\begin{proof}
    See Appendix~\ref{app:proof_gauge_Csep}.
\end{proof}

\begin{remark}\label{rem:gauge_sep_optimizer}
    When $\gammasep(C)$ is finite, the supremum in the duality relation
    \begin{equation}
        \gammasep(C) = \sup_{Z \neq 0} \frac{\Tr[Z^\top C]}{\phisep(Z)}
    \end{equation}
    is attained as a maximum.
    Let $A^{+} C\, (B^\top)^{+} = U \Sigma V^\top$ be a singular value decomposition.
    Then the maximum is achieved at $Z_{\star} = (A^\top)^{+}\, U V^\top\, B^{+}$.
\end{remark}

\subsection{Gauge function of $\Cquantum$}
\begin{theorem}[{Gauge function of $\Cquantum$}]\label{thm:gauge_function_of_Cquantum}
    Let $r := \min\{\rank A, \rank B\}$.

    When $m \geq 3$, the gauge function $\gammaquantum$ of $\Cquantum$ is given by
    \begin{equation}
        \gammaquantum (C) = \begin{cases}
            s'_1 + s'_2 + \epsilon' s'_3  & \text{if~} r=3,~ \ran C \subset \ran A,~\ran C^\top \subset \ran B, \\
            \gammamax (C) & \text{if~} r \leq 2,
        \end{cases}
    \end{equation}
    where $\epsilon' := \op{sgn} (\det (A^{+} C\, (B^\top)^{+}))$, and $\gammamax$ is given in Theorem~\ref{thm:gauge_function_of_Cmax}.

    When $m=2$, this specializes to
    \begin{equation}
        \gammaquantum (C) = \frac{1}{2} \left(\sqrt{\Tr[L_{\alpha}\, C\, L_{\beta}^\top\, C^\top]} + \sqrt{\Tr[L_{\alpha}\, C\, L_{-\beta}^\top\, C^\top]}\right),
        \label{eq:gauge_function_of_Cquantum_matrix_form}
    \end{equation}
    where
    \begin{equation}
        L_{\theta} := \frac{1}{\sin^2\theta}
        \begin{pmatrix}
            1 & -\ep{\im\theta} \\
            -\ep{-\im\theta} & 1
        \end{pmatrix},
    \end{equation}
    or equivalently,
    \begin{equation}
        \gammaquantum (C) = \frac{1}{2} \left(\sqrt{\bs{c}^\top F_{\alpha,\beta} \bs{c}} + \sqrt{\bs{c}^\top F_{\alpha,-\beta} \bs{c}}\right),
    \end{equation}
    where
    \begin{equation}
        F_{\alpha,\beta} := \frac{1}{\sin^2\alpha \sin^2\beta}
                \begin{pmatrix}
                1 & -\cos\alpha & -\cos\beta & \cos(\alpha+\beta)\\
                -\cos\alpha & 1 & \cos(\alpha-\beta) & -\cos\beta\\
                -\cos\beta & \cos(\alpha-\beta) & 1 & -\cos\alpha\\
                \cos(\alpha+\beta) & -\cos\beta & -\cos\alpha & 1
                \end{pmatrix}.
    \end{equation}
\end{theorem}
\begin{proof}
    See Appendix~\ref{app:proof_gauge_Cquantum}.
\end{proof}

\begin{remark}\label{rem:gauge_qm_optimizer}
    When $r = 3$ and $\gammaquantum(C)$ is finite, the supremum in the duality relation
    \begin{equation}
        \gammaquantum(C) = \sup_{Z \neq 0} \frac{\Tr[Z^\top C]}{\phiquantum(Z)}
    \end{equation}
    is attained as a maximum.
    Let $A^{+} C\, (B^\top)^{+} = U \Sigma V^\top$ be a singular value decomposition.
    Then the maximum is achieved at $Z_{\star} = (A^\top)^{+}\, U\, \op{diag}(1,1,\eta)\, V^\top\, B^{+}$, where $\eta := \op{sgn}(\det(U V^\top))$.
    When $r \leq 2$, the result reduces to the case of $\Cmax$ (Remark~\ref{rem:gauge_max_optimizer}).
\end{remark}

\subsection{Gauge function of $\Cmax$}
\begin{theorem}[{Gauge function of $\Cmax$}]\label{thm:gauge_function_of_Cmax}
    The gauge function $\gammamax$ of $\Cmax$ is given by
    \begin{equation}
        \gammamax (C) = \begin{cases}
            \opnorm{ A^{+} C\, (B^\top)^{+}} = s'_1, & \text{if~} \ran C \subset \ran A,~ \ran C^\top \subset \ran B, \\
            +\infty, & \text{otherwise}.
        \end{cases}
    \end{equation}
\end{theorem}
\begin{proof}
    See Appendix~\ref{app:proof_gauge_Cmax}.
\end{proof}

\begin{remark}\label{rem:gauge_max_optimizer}
    When $\gammamax(C)$ is finite, the supremum in the duality relation
    \begin{equation}
        \gammamax(C) = \sup_{Z \neq 0} \frac{\Tr[Z^\top C]}{\phimax(Z)}
    \end{equation}
    is attained as a maximum.
    Let $A^{+} C\, (B^\top)^{+} = U \Sigma V^\top$ be a singular value decomposition, and let $\bs{u}_1, \bs{v}_1$ denote the first columns of $U, V$ (corresponding to the largest singular value $s'_1$).
    Then the maximum is achieved at $Z_{\star} = (A^\top)^{+}\, \bs{u}_1 \bs{v}_1^\top\, B^{+}$.
\end{remark}

\section{Convex hull characterizations of the correlation sets}\label{sec:convex_hull_characterization}

The support function theorems (Theorems~\ref{thm:support_function_of_Csep}--\ref{thm:support_function_of_Cmax}) yield the following geometric characterizations of the three correlation sets.
We define
\begin{align}
    \mrm{O}(n) &:= \{Q \in \mcal{M}_n(\mbb{R}) \,:\, Q^\top Q = I\}, \\
    \mrm{SO}(n) &:= \{Q \in \mrm{O}(n) \,:\,  \det Q = +1\},\\
    \mrm{SO}^{-}(n) &:= \{Q \in \mrm{O}(n) \,:\,  \det Q = -1\}.
\end{align}

\begin{theorem}[{Convex hull characterization of $\Csep$}]{\label{thm:convex_hull_characterization_of_Csep}}
    \begin{equation}
        \Csep = \op{Conv} \{A\, \bs{r}_A\bs{r}_B^\top\, B^{\top} \,:\, \bs{r}_A, \bs{r}_B \in \mbb{R}^3,\, \|\bs{r}_A\| = \|\bs{r}_B\| = 1\}.
    \end{equation}
\end{theorem}
\begin{proof}
    See Appendix~\ref{app:proof_convex_hull_Csep}.
\end{proof}

\begin{theorem}[{Convex hull characterization of $\Cquantum$}]{\label{thm:convex_hull_characterization_of_Cquantum}}
    \begin{equation}
        \Cquantum = \op{Conv} \{A Q B^{\top} \,:\, Q \in \mrm{SO}^{-}(3)\}.
    \end{equation}
\end{theorem}
\begin{proof}
    See Appendix~\ref{app:proof_convex_hull_Cquantum}.
\end{proof}

\begin{theorem}[{Convex hull characterization of $\Cmax$}]{\label{thm:convex_hull_characterzation_of_Cmax}}
    \begin{equation}
        \Cmax = \op{Conv} \{A Q B^{\top} \,:\, Q \in \mrm{O}(3)\}.
    \end{equation}
\end{theorem}
\begin{proof}
    See Appendix~\ref{app:proof_convex_hull_Cmax}.
\end{proof}

\begin{remark}[Physical states realizing extreme points]\label{rem:extreme_point_states}

We write a general two-qubit self-adjoint operator in the Pauli basis as
\begin{equation}\label{eq:pauli_expansion}
    \hat{\rho}
    = \frac{1}{4}\biggl(
        \hat{I}_4
        + \sum_{i=1}^{3} r_{A,i}\, \hat{\sigma}_i \otimes \hat{I}_2
        + \sum_{j=1}^{3} r_{B,j}\, \hat{I}_2 \otimes \hat{\sigma}_j
        + \sum_{i,j=1}^{3} T_{ij}\, \hat{\sigma}_i \otimes \hat{\sigma}_j
    \biggr),
\end{equation}
where $\bs{r}_A, \bs{r}_B \in \mbb{R}^3$ are the local Bloch vectors and $T = (T_{ij}) \in \mcal{M}_3(\mbb{R})$ is the correlation matrix in the Pauli basis.

    The extreme points of each correlation set are realized by specific classes of two-qubit states:
    \begin{enumerate}
        \item \emph{Separable correlation set.}
        A pure product state with Bloch vectors $\bs{r}_A, \bs{r}_B \in \mbb{R}^3$ satisfying $\|\bs{r}_A\| = \|\bs{r}_B\| = 1$ has the density operator
        \begin{equation}
            \hat{\rho} = \frac{1}{2}(\hat{I}_2 + \bs{r}_A \cdot \hat{\bs{\sigma}}) \otimes \frac{1}{2}(\hat{I}_2 + \bs{r}_B \cdot \hat{\bs{\sigma}})
            = \frac{1}{4}\biggl(\hat{I}_4 + \sum_{i=1}^{3} r_{A,i}\, \hat{\sigma}_i \otimes \hat{I}_2 + \sum_{j=1}^{3} r_{B,j}\, \hat{I}_2 \otimes \hat{\sigma}_j + \sum_{i,j=1}^{3} r_{A,i}\, r_{B,j}\, \hat{\sigma}_i \otimes \hat{\sigma}_j\biggr).
        \end{equation}
        The Pauli-basis correlation matrix is therefore $T_{ij} = r_{A,i}\, r_{B,j}$, i.e., $T = \bs{r}_A \bs{r}_B^\top$, which gives the correlation matrix $C = A\bs{r}_A \bs{r}_B^\top B^\top$.
        Therefore, pure product states give the extreme points $A\bs{r}_A\bs{r}_B^\top B^\top$ of $\Csep$.

        \item \emph{Quantum correlation set.}
        Every maximally entangled two-qubit state can be written, up to a global phase, as $(\hat{I}_2 \otimes \hat{U}){|\Phi^+\ket}$ for some $\hat{U} \in \mrm{SU}(2)$.
        Writing $R := \op{diag}(1,-1,1)$ for the correlation tensor of ${|\Phi^+\ket}$, the density operator of this state is
        \begin{equation}
            (\hat{I}_2 \otimes \hat{U}){|\Phi^+\ket}\!\bra{\Phi^+}|(\hat{I}_2 \otimes \hat{U}^\dagger)
            = \frac{1}{4}\biggl(\hat{I}_4 + \sum_{i,j=1}^{3} R_{ij}\, \hat{\sigma}_i \otimes \hat{U}\hat{\sigma}_j\hat{U}^\dagger\biggr).
        \end{equation}
        By the $\mrm{SU}(2) \to \mrm{SO}(3)$ homomorphism, there exists $Q^U \in \mrm{SO}(3)$ such that $\hat{U}\hat{\sigma}_j\hat{U}^\dagger = \sum_{k=1}^{3} Q^U_{jk}\, \hat{\sigma}_k$.
        Substituting this gives
        \begin{equation}
            \hat{\rho} = \frac{1}{4}\biggl(\hat{I}_4 + \sum_{i,k=1}^{3} Q_{ik}\, \hat{\sigma}_i \otimes \hat{\sigma}_k\biggr), \quad Q := R\, Q^U.
        \end{equation}
        Since $R \in \mrm{SO}^{-}(3)$ and $Q^U \in \mrm{SO}(3)$, we have $Q \in \mrm{SO}^{-}(3)$.
        Conversely, every $Q \in \mrm{SO}^{-}(3)$ arises this way by choosing $Q^U = R^{-1} Q \in \mrm{SO}(3)$.
        Therefore, maximally entangled states take the form
        \begin{equation}
            \hat{\rho} = \frac{1}{4}\biggl(\hat{I}_4 + \sum_{i,j=1}^{3} Q_{ij}\, \hat{\sigma}_i \otimes \hat{\sigma}_j\biggr), \quad Q \in \mrm{SO}^{-}(3),
        \end{equation}
        and produce the correlation matrix $C = AQB^\top$.
        Conversely, by Proposition~\ref{prop:rigidity} below, maximally entangled states are the \emph{only} quantum states that realize these correlations: any density operator with an orthogonal correlation matrix $T \in \mrm{O}(3)$ must have vanishing local Bloch vectors $\bs{r}_A = \bs{r}_B = \bs{0}$, and is therefore a maximally entangled pure state when $\det T = -1$.

        \item \emph{Maximal correlation set.}
        Similarly, states of the form $\hat{\rho} = \frac{1}{4}(\hat{I}_4 + \sum_{i,j} Q_{ij}\, \hat{\sigma}_i \otimes \hat{\sigma}_j)$ with $Q \in \mrm{O}(3)$ produce $C = AQB^\top$.
        When $\det Q = +1$, such a state is not a valid quantum state in general (it may fail positivity), but it belongs to the maximal tensor product state space $\Smax$.
        Conversely, again by Proposition~\ref{prop:rigidity}, these are the \emph{only} block-positive states that realize the corresponding correlations: any block-positive state with $T \in \mrm{SO}(3)$ must have $\bs{r}_A = \bs{r}_B = \bs{0}$.
    \end{enumerate}
\end{remark}

\begin{proposition}[Rigidity of extremal correlations]\label{prop:rigidity}
    Let $\hat{\rho}$ be a normalized block-positive two-qubit operator (i.e., $\hat{\rho} \in \Smax$) with Pauli expansion~\eqref{eq:pauli_expansion}.
    If $T \in \mrm{O}(3)$, then $\bs{r}_A = \bs{r}_B = \bs{0}$.
    In particular, any block-positive state with an orthogonal correlation matrix is uniquely determined by $T$ and takes the form
    \begin{equation}
        \hat{\rho} = \frac{1}{4}\biggl(\hat{I}_4 + \sum_{i,j=1}^{3} T_{ij}\, \hat{\sigma}_i \otimes \hat{\sigma}_j\biggr).
    \end{equation}
\end{proposition}
\begin{proof}
    See Appendix~\ref{app:rigidity}.
\end{proof}

\section{Detection of Entanglement}\label{sec:detection_of_entanglement}
\subsection{Entanglement detection via the support function $\phisep$}

The support function $\phisep$ provides a direct criterion for entanglement detection.
If a correlation matrix $C$ satisfies $\Tr[Z^\top C] > \phisep(Z)$ for some $Z \in \MmR$, then $C \notin \Csep$, and hence the underlying quantum state must be entangled.

This criterion admits a natural operator-theoretic formulation.
Define the Bell operator $\hat{S}_{AB}(Z) := \sum_{i,j} z_{ij}\, \hat{A}_i \otimes \hat{B}_j$ and the self-adjoint operator
\begin{equation}
    \hat{W}_{\mrm{ent}}(Z) := \hat{I}_4 - \frac{\hat{S}_{AB}(Z)}{\phisep(Z)}.
\end{equation}
For any separable state $\hat{\rho} \in \Ssep$, the correlation matrix $C = \Phi_{AB}(\hat{\rho})$ satisfies $\Tr[Z^\top C] \leq \phisep(Z)$ by definition of the support function, so
\begin{equation}
    \Tr[\hat{\rho}\, \hat{W}_{\mrm{ent}}(Z)] = 1 - \frac{\Tr[\hat{\rho}\, \hat{S}_{AB}(Z)]}{\phisep(Z)} = 1 - \frac{\Tr[Z^\top C]}{\phisep(Z)} \geq 0.
    \label{eq:entanglement_witness}
\end{equation}
As shown in the following subsection, when $\rank A, \rank B \geq 2$, there exist a matrix $Z \in \MmR$ and an entangled state $\hat{\rho} \in \Squantum$ for which $\Tr[\hat{\rho}\, \hat{W}_{\mrm{ent}}(Z)] < 0$. Therefore, $\hat{W}_{\mrm{ent}}(Z)$ serves as an entanglement witness~\cite{terhal2000bell,lewenstein2000optimization}.

\subsection{The gauge function $\gammasep$ as a robustness measure for entanglement detection}
Equation~\eqref{eq:entanglement_witness} shows that entanglement is detected whenever the ratio $\Tr[Z^\top C]/\phisep (Z)$ exceeds unity; the larger this ratio, the stronger the detection signal.
In an experimental setting, the coefficient matrix $Z$ is not a physical parameter but a weight assigned during statistical post-processing of the measured correlations, and can therefore be freely optimized.
We thus define the detection sensitivity for a given correlation matrix $C$ as the supremum of this ratio over all nonzero $Z$.
The duality relation~\eqref{eq:gauge_function_as_a_dual_of_support_function} then identifies the detection sensitivity with the gauge function of $\Csep$:
\begin{proposition}\label{thm:state_dependent_detection}
    For any $C \in \MmR$,
    \begin{equation}
        \sup_{Z \neq 0} \frac{\Tr[Z^\top C]}{\phisep (Z)}
        = \gammasep(C).
        \label{eq:state-dependent detection sensitivity}
    \end{equation}
\end{proposition}
\begin{remark}\label{rem:max_attained}
    The duality relation~\eqref{eq:gauge_function_as_a_dual_of_support_function} gives the right-hand side as a supremum.
    When $\gammasep(C) < +\infty$ (equivalently, $\ran C \subset \ran A$ and $\ran C^\top \subset \ran B$), the supremum is in fact attained, which justifies writing $\max$ in~\eqref{eq:state-dependent detection sensitivity} in this case.
    To see this, note that the numerator $\Tr[Z^\top C]$ and the denominator $\phisep(Z)$ are both positively homogeneous of degree one in $Z$, so the ratio is positively homogeneous of degree zero: replacing $Z$ by $\lambda Z$ ($\lambda > 0$) leaves the ratio unchanged.
    The supremum over $Z \neq 0$ therefore equals the supremum over the unit sphere $\{Z \in \MmR : \|Z\|_{\HS} = 1\}$, which is compact.
    Since $\phisep$ is continuous by~\eqref{eq:support_function_lipschitz} and positive on the unit sphere (because $\Csep$ has nonempty interior), the ratio $\Tr[Z^\top C]/\phisep(Z)$ is continuous on this compact set, and therefore attains its supremum by the extreme value theorem.
\end{remark}

The gauge function $\gammasep(C)$ also admits a transparent interpretation as a noise-robustness measure.
By definition, $\gammasep(C) = \inf\{t > 0 : C/t \in \Csep\}$, so $C/\gammasep(C)$ lies on the boundary of $\Csep$. Since the completely mixed state $\hat{I}/4$ produces the zero correlation matrix, mixing a state with the completely mixed state simply rescales the correlations:
\begin{equation}
    (1-p)\, C + p \cdot 0 = (1-p)\, C.
\end{equation}
The rescaled correlation $(1-p)\,C$ belongs to $\Csep$ if and only if $(1-p) \leq 1/\gammasep(C)$, i.e., $p \geq 1 - 1/\gammasep(C)$. Therefore, entanglement remains detectable through correlation measurements as long as the noise fraction satisfies $p < 1 - 1/\gammasep(C)$. In this sense, $\gammasep(C)$ quantifies the tolerance of the correlation $C$ to depolarizing noise: the larger $\gammasep(C)$, the more noise $C$ can withstand before entanglement becomes undetectable.

The necessary and sufficient condition $\gammasep(C) \leq 1$ for $C \in \Csep$ implies that for any separable density operator $\hat{\rho} \in \Ssep$, the correlation $C = \Phi_{AB}(\hat{\rho})$ must satisfy
\begin{equation}
    \trnorm{A^{+} C\, (B^\top)^{+}} \leq 1.
    \label{eq:CCN_analog}
\end{equation}
This bears a resemblance to the computable cross-norm
criterion of Rudolph~\cite{Rudolph2004}, which states that
any separable density operator $\hat{\rho} \in \Ssep$ must satisfy
$\|\mathfrak{A}(\hat{\rho})\|_1 \leq 1$, where $\mathfrak{A}$ denotes
the realignment map.

\subsection{Fundamental limit of detection sensitivity}

When arbitrary quantum states may be prepared, the fundamental limit of detection sensitivity under fixed measurement settings $A$ and $B$ is obtained by maximizing $\gammasep$ over all quantum-realizable correlations:
\begin{equation}
    \sup_{C \in \Cquantum} \gammasep(C)
    = \sup_{C \in \Cquantum} \max_{Z \neq 0} \frac{\Tr[Z^\top C]}{\phisep(Z)}
    = \max_{Z \neq 0} \frac{\phiquantum(Z)}{\phisep(Z)}.
\end{equation}
\begin{theorem}\label{thm:maximum_ratio_of_support_functions}
    Let $r := \min\{\rank A, \rank B\}$. Then
    \begin{equation}
        \max_{Z \neq 0} \frac{\phiquantum(Z)}{\phisep(Z)} = \begin{cases}
            3 & \text{if~} r=3, \\
            2 & \text{if~} r=2, \\
            1 & \text{if~} r=1.
        \end{cases}
    \end{equation}
\end{theorem}
\begin{proof}
    See Appendix~\ref{app:proof_maximum_ratio_sep}.
\end{proof}
This supremum coincides with the containment radius of $\Cquantum$ with respect to $\Csep$. For two convex compact sets $K$ and $K'$ whose interiors contain the origin, the containment radius is defined as
\begin{equation}
    R(K, K') := \inf \{\lambda > 0 \,:\, K \subset \lambda K'\} = \sup_{C \in K} \gamma_{K'}(C) = \max_{Z \neq 0} \frac{\phi_K(Z)}{\phi_{K'}(Z)}.
\end{equation}
In particular, since $\Csep$ is closed, $R(\Cquantum, \Csep) = 3$ when $r = 3$ means that $\Cquantum \subset 3\, \Csep$ and that this scaling is optimal.

The following theorem characterizes the correlations $C \in \Cquantum$ that achieve the maximum.
\begin{theorem}\label{thm:maximizer_gammasep}
    Let $r := \min\{\rank A, \rank B\}$. The following correlations $C \in \Cquantum$ achieve $\sup_{C \in \Cquantum} \gammasep(C)$:
    \begin{enumerate}
        \item When $r = 3$: $C = A Q B^\top$ for any $Q \in \mrm{SO}^{-}(3)$.
        \item When $r = 2$:
        \begin{itemize}
            \item If $\rank A = 3$ or $\rank B = 3$: $C = A Q B^\top$ for any $Q \in \mrm{SO}^{-}(3)$.
            \item If $\rank A = \rank B = 2$: $C = A Q B^\top$ for $Q \in \mrm{SO}^{-}(3)$ satisfying $Q(\ran B^\top) = \ran A^\top$.
        \end{itemize}
    \end{enumerate}
\end{theorem}
\begin{proof}
    See Appendix~\ref{app:proof_maximizer_gammasep}.
\end{proof}

\begin{remark}\label{rem:maximizer_gammasep_physical}
    By Remark~\ref{rem:extreme_point_states}, the extreme points $AQB^\top$ with $Q \in \mrm{SO}^{-}(3)$ of $\Cquantum$ are realized by maximally entangled states.
    Moreover, by Proposition~\ref{prop:rigidity}, maximally entangled states are the \emph{only} quantum states that produce these correlations: no state with nonzero local Bloch vectors can have the same correlation matrix.
    Therefore, the correlations that maximize the entanglement detection sensitivity $\gammasep$ are exactly those produced by maximally entangled states.

    The coefficient matrix $Z$ that achieves the maximum ratio $\phiquantum(Z)/\phisep(Z)$ is determined by such a maximizing correlation $C$ via the correspondence in Remark~\ref{rem:gauge_sep_optimizer}.
    For instance, when $r = 3$, the maximizer $C_{\star} = A Q B^\top$ with $Q \in \mrm{SO}^{-}(3)$ gives $A^{+} C_{\star} (B^\top)^{+} = Q$, whose singular value decomposition is $Q = U I_3 V^\top$.
    By Remark~\ref{rem:gauge_sep_optimizer}, the corresponding optimal coefficient matrix is $Z_{\star} = (A^\top)^{+}\, U V^\top\, B^{+}$, for which $A^\top Z_{\star} B = Q$ has singular values $s_1 = s_2 = s_3 = 1$ and $\det(A^\top Z_{\star} B) = -1$, yielding $\phiquantum(Z_{\star})/\phisep(Z_{\star}) = 3$.
    Since $U V^\top = A^{+} C_{\star} (B^\top)^{+}$ when all singular values equal~$1$, the identity $(A^\top)^{+} A^{+} = (A A^\top)^{+}$ gives a compact expression in terms of the Gram matrices $G_A = A A^\top$ and $G_B = B B^\top$:
    \begin{equation}
        Z_{\star} = G_A^{-1}\, C_{\star} G_B^{-1}.
    \end{equation}
\end{remark}

\subsection{Noise tolerance: Werner states}\label{subsec:noise_tolerance_Werner}

We apply the preceding results to the entanglement detection of Werner states~\cite{werner1989quantum} under depolarizing noise.
The Werner state with noise fraction $p \in [0,1]$ is
\begin{equation}
    \hat{\rho}_p := (1-p) \, |\Phi^+\ket\bra\Phi^+| + p\,\frac{\hat{I}_4}{4},
    \label{eq:Werner_state}
\end{equation}
where $|\Phi^+\ket := (|00\ket + |11\ket)/\sqrt{2}$.
For two-qubit systems, the PPT criterion~\cite{peres1996separability} is necessary and sufficient for separability~\cite{horodecki1996separability}: $\hat{\rho}_p$ is entangled if and only if $p < 2/3$.

Since the completely mixed state produces the zero correlation matrix, we have $\Phi_{AB}(\hat{\rho}_p) = (1-p)\, C_{\Phi^+}$, where $C_{\Phi^+} := \Phi_{AB}(|\Phi^+\ket\bra\Phi^+|)$.
By the positive homogeneity of the gauge function,
\begin{equation}
    \gammasep\bigl((1-p)\, C_{\Phi^+}\bigr) = (1-p) \, \gammasep(C_{\Phi^+}),
\end{equation}
and entanglement is detected if and only if this exceeds unity, i.e.,
\begin{equation}\label{eq:Werner_p_crit}
    p < p_{\mrm{crit}} := 1 - \frac{1}{\gammasep(C_{\Phi^+})}.
\end{equation}

The state $|\Phi^+\ket$ is a maximally entangled state whose correlation matrix in the Pauli basis is $T_{\Phi^{+}} := \op{diag}(1,-1,1) \in \mrm{SO}^-(3)$.
By Theorem~\ref{thm:maximizer_gammasep}, when $r := \min\{\rank A, \rank B\} \geq 2$, the maximally entangled state achieves the maximum of $\gammasep$ over $\Cquantum$:
\begin{equation}
    \gammasep(C_{\Phi^+}) = r.
    \label{eq:gamma_sep_maximally_entangled}
\end{equation}
Therefore,
\begin{equation}
    p_{\mrm{crit}} = 1 - \frac{1}{r}.
    \label{eq:p_crit_general}
\end{equation}
When $r = 3$, $p_{\mrm{crit}} = 2/3$, which coincides with the PPT threshold for Werner states; the PPT criterion is necessary and sufficient for separability in $2 \times 2$ systems~\cite{horodecki1996separability}.
The optimal coefficient matrix is given by Remark~\ref{rem:maximizer_gammasep_physical}:
\begin{equation}\label{eq:Z_star_Werner}
    Z_{\star} = G_A^{-1}\, C_{\Phi^+}\, G_B^{-1},
\end{equation}
with $\phisep(Z_{\star}) = \opnorm{T_{\Phi^{+}}} = 1$, so the entanglement witness~\eqref{eq:entanglement_witness} becomes $\hat{W}_{\mrm{ent}}(Z_{\star}) = \hat{I}_4 - \hat{S}_{AB}(Z_{\star})$ and detects entanglement when $p < 1 - 1/r$.

We now compare the noise tolerance of our gauge-function criterion with device-independent detection via Bell inequalities in concrete measurement settings.

\subsubsection{The $(2,2,2)$ scenario}

Consider the standard CHSH measurement settings:
\begin{equation}
    \hat{A}_1 = \hat{\sigma}_3, \quad \hat{A}_2 = \hat{\sigma}_1, \quad
    \hat{B}_1 = \frac{\hat{\sigma}_3 + \hat{\sigma}_1}{\sqrt{2}}, \quad
    \hat{B}_2 = \frac{\hat{\sigma}_3 - \hat{\sigma}_1}{\sqrt{2}}.
\end{equation}
Since $m = 2$, we have $r = \min\{\rank A, \rank B\} = 2$, so $\gammasep(C_{\Phi^+}) = 2$ and $p_{\mrm{crit}} = 1/2$.

For comparison, the CHSH inequality~\cite{clauser1969proposed} provides a device-independent entanglement test.
The Tsirel'son bound~\cite{cirel1980quantum} gives the maximum quantum value $2\sqrt{2}$ of the CHSH expression, achieved by the above settings with the maximally entangled state.
For the Werner state~\eqref{eq:Werner_state}, the CHSH value scales as $(1-p) \cdot 2\sqrt{2}$, and violation of the local bound~$2$ requires
\begin{equation}
    p < 1 - \frac{1}{\sqrt{2}} \approx 0.293.
\end{equation}
The gauge-function criterion thus tolerates more noise ($p_{\mrm{crit}} = 0.5$ vs.\ $\approx 0.293$), reflecting the advantage of device-dependent over device-independent detection.

\subsubsection{The $(2,3,2)$ scenario}\label{subsubsec:noise_tolerance_232}

Consider the measurement settings $A = I_3$ (i.e., $\hat{A}_i = \hat{\sigma}_i$) and
\begin{equation}
    B = \begin{pmatrix}
        \frac{1}{\sqrt{2}} & 0 & \frac{1}{\sqrt{2}} \\[3pt]
        0 & 1 & 0 \\[3pt]
        \frac{1}{\sqrt{2}} & 0 & -\frac{1}{\sqrt{2}}
    \end{pmatrix}.
\end{equation}
Since $A$ and $B$ both have full rank, $r = 3$.
The maximally entangled state $|\Phi^+\ket$ produces the correlation $C_{\Phi^+} = A T_{\Phi^{+}} B^\top = T_{\Phi^{+}} B^\top$ with $T_{\Phi^{+}} = \op{diag}(1,-1,1) \in \mrm{SO}^{-}(3)$.
By~\eqref{eq:gamma_sep_maximally_entangled}, $\gammasep(C_{\Phi^+}) = 3$ and $p_{\mrm{crit}} = 2/3$, saturating the PPT bound.
Since $B$ is orthogonal, $G_A = G_B = I_3$, and the optimal coefficient matrix~\eqref{eq:Z_star_Werner} reduces to
\begin{equation}
    Z_{\star} = C_{\Phi^+} = \begin{pmatrix}
        \frac{1}{\sqrt{2}} & 0 & \frac{1}{\sqrt{2}} \\[3pt]
        0 & -1 & 0 \\[3pt]
        \frac{1}{\sqrt{2}} & 0 & -\frac{1}{\sqrt{2}}
    \end{pmatrix}.
\end{equation}

In the $(2,3,2)$ scenario, a device-independent entanglement test can be performed by applying the CHSH inequality to a subset of two measurement settings per party (e.g., $\hat{A}_1, \hat{A}_3$ and $\hat{B}_1, \hat{B}_3$), yielding the same critical noise as in the $(2,2,2)$ scenario:
\begin{equation}
    p_{\mrm{crit}}^{\mrm{CHSH}} = 1 - \frac{1}{\sqrt{2}} \approx 0.293.
\end{equation}
The $I_{3322}$ inequality~\cite{collins2004relevant}---the unique tight Bell inequality for the $(2,3,2)$ scenario that is inequivalent to the CHSH inequality---provides another device-independent test within this scenario.
The quantum maximum of $I_{3322}$ equals $1/4$~\cite{collins2004relevant}, achieved by $|\Phi^+\ket$ with equiangular in-plane measurement settings.
The optimal measurement matrices are
\begin{equation}\label{eq:I3322_optimal_settings}
    A = \begin{pmatrix} 0 & 0 & 1 \\[3pt]
        -\frac{\sqrt{3}}{2} & 0 & \frac{1}{2} \\[3pt]
        -\frac{\sqrt{3}}{2} & 0 & -\frac{1}{2} \end{pmatrix}, \qquad
    B = \begin{pmatrix} -\frac{\sqrt{3}}{2} & 0 & \frac{1}{2} \\[3pt]
        0 & 0 & 1 \\[3pt]
        \frac{\sqrt{3}}{2} & 0 & \frac{1}{2} \end{pmatrix},
\end{equation}
whose rows are unit vectors in the $xz$-plane at angles $0, -\pi/3, -2\pi/3$ (Alice) and $-\pi/3, 0, \pi/3$ (Bob) from the $z$-axis, equally spaced by $\pi/3$.
For the completely mixed state $\hat{I}_4/4$, a direct calculation gives $I_{3322}(\hat{I}_4/4) = -1$.
By linearity, the Werner state~\eqref{eq:Werner_state} satisfies
\begin{equation}
    I_{3322}(\hat{\rho}_p) = (1-p)\cdot\frac{1}{4} + p\cdot(-1) = \frac{1}{4} - \frac{5p}{4}.
\end{equation}
Violation of the local-realistic bound $I_{3322} \leq 0$ requires $p < 1/5$, so
\begin{equation}\label{eq:I3322_threshold}
    p_{\mrm{crit}}^{I_{3322}} = \frac{1}{5} = 0.2.
\end{equation}
Despite being tailored for the $(2,3,2)$ scenario, $I_{3322}$ yields a weaker noise threshold than CHSH ($1/5 < 1 - 1/\sqrt{2} \approx 0.293$).

By contrast, the gauge-function criterion achieves $p_{\mrm{crit}} = 2/3 \approx 0.667$, which is optimal.

The measurement setting $A = B = I_3$, i.e., $\hat{A}_i = \hat{B}_i = \hat{\sigma}_i$ for $i = 1, 2, 3$, further simplifies the detection protocol.
Since $G_A = G_B = I_3$, the optimal coefficient matrix~\eqref{eq:Z_star_Werner} reduces to
\begin{equation}
    Z_{\star} = C_{\Phi^+} = T_{\Phi^{+}} = \op{diag}(1, -1, 1).
\end{equation}
This matrix is diagonal, so the Bell operator involves only the three measurement combinations $\hat{\sigma}_i \otimes \hat{\sigma}_i$:
\begin{equation}
    \hat{S}_{AB}(Z_{\star}) = \hat{\sigma}_1 \otimes \hat{\sigma}_1 - \hat{\sigma}_2 \otimes \hat{\sigma}_2 + \hat{\sigma}_3 \otimes \hat{\sigma}_3.
\end{equation}
Therefore, among the $3^2 = 9$ possible measurement combinations in the $(2,3,2)$ scenario, only 3 are needed for optimal entanglement detection.
The resulting criterion detects entanglement whenever $p < 2/3$, matching the PPT bound with far fewer measurements than full state tomography.

\begin{remark}
    The critical noise $p_{\mrm{crit}} = 1 - 1/r$ depends only on $r = \min\{\rank A, \rank B\}$ and not on the specific measurement directions.
    When $r = 3$ (any non-coplanar setting with $m \geq 3$), $p_{\mrm{crit}} = 2/3$ saturates the PPT bound.
    When $r = 2$---which includes the CHSH scenario $(m = 2)$ as well as coplanar $(2,3,2)$ settings---$p_{\mrm{crit}} = 1/2$.
    When $r = 1$, $p_{\mrm{crit}} = 0$, i.e., no entanglement can be detected.
\end{remark}

Table~\ref{tab:noise_tolerance_entanglement} summarizes the critical noise fractions for entanglement detection of Werner states across the methods discussed above.

\begin{table}[t]
    \caption{Critical noise $p_{\mrm{crit}}$ for entanglement detection of the Werner state $\hat{\rho}_p = (1-p)\,|\Phi^+\ket\bra\Phi^+| + p\,\hat{I}_4/4$.
    Entanglement is detected when $p < p_{\mrm{crit}}$.
    The PPT criterion~\cite{peres1996separability,horodecki1996separability} is necessary and sufficient for two-qubit systems but requires full state tomography;
    it is inapplicable under the limited measurement settings of the $(2,2,2)$ scenario.
    The $I_{3322}$ inequality~\cite{collins2004relevant} requires three measurement settings per party and is not applicable in the $(2,2,2)$ scenario (dash).}
    \label{tab:noise_tolerance_entanglement}
    \begin{ruledtabular}
    \begin{tabular}{lcc}
        Method & $(2,2,2)$ & $(2,3,2)$ \\
        \hline
        PPT criterion (tomography) & --- & $2/3 \approx 0.667$ \\
        Gauge function $\gammasep$ (this work) & $1/2 = 0.5$ & $2/3 \approx 0.667$ \\
        CHSH inequality~\cite{clauser1969proposed} & $1 - 1/\sqrt{2} \approx 0.293$ & $1 - 1/\sqrt{2} \approx 0.293$ \\
        $I_{3322}$ inequality~\cite{collins2004relevant} & --- & $1/5 = 0.2$ \\
    \end{tabular}
    \end{ruledtabular}
\end{table}

\section{Detection of Beyond-Quantum States}\label{sec:detection_of_BQS}
\subsection{Beyond-quantum state detection via the support function $\phiquantum$}
Analogously to the entanglement case, the support function $\phiquantum$ provides a direct criterion for beyond-quantum state detection.
If a correlation matrix $C$ satisfies $\Tr[Z^\top C] > \phiquantum(Z)$ for some $Z \in \MmR$, then $C \notin \Cquantum$, and hence the underlying state must be beyond-quantum.

This criterion has an analogous operator-theoretic formulation. Define the self-adjoint operator
\begin{equation}
    \hat{W}_{\mrm{bqs}}(Z) := \hat{I}_4 - \frac{\hat{S}_{AB}(Z)}{\phiquantum (Z)}.
    \label{eq:beyond-quantum_state_witness}
\end{equation}
For any quantum state $\hat{\rho} \in \Squantum$, the correlation matrix $C = \Phi_{AB}(\hat{\rho})$ satisfies $\Tr[Z^\top C] \leq \phiquantum(Z)$ by definition of the support function, so
\begin{equation}
    \Tr[\hat{\rho} \, \hat{W}_{\mrm{bqs}}(Z)]
    = 1 - \frac{\Tr[\hat{\rho}\, \hat{S}_{AB}(Z)]}{\phiquantum(Z)}
    = 1 - \frac{\Tr[Z^\top C]}{\phiquantum(Z)} \geq 0.
    \label{eq:beyond_quantum_detection}
\end{equation}
As shown in the following subsection, when $\rank A,\rank B \geq 2$, there exist a matrix $Z \in \MmR$ and a beyond-quantum state $\hat{\rho} \in \Smax$ for which $\Tr[\hat{\rho} \, \hat{W}_{\mrm{bqs}}(Z)] < 0$.
Therefore, $\hat{W}_{\mrm{bqs}}(Z)$ serves as an operator witnessing beyond-quantum states~\cite{arai2024detecting}.

\subsection{The gauge function $\gammaquantum$ as a robustness measure for beyond-quantum state detection}

Equation~\eqref{eq:beyond_quantum_detection} shows that a beyond-quantum state is detected whenever the ratio $\Tr[Z^\top C]/\phiquantum(Z)$ exceeds unity; the larger this ratio, the stronger the detection signal.
As in the entanglement case, the coefficient matrix $Z$ can be freely optimized during statistical post-processing, so we define the detection sensitivity for beyond-quantum correlations as the supremum of this ratio over all nonzero $Z$.
The duality relation~\eqref{eq:gauge_function_as_a_dual_of_support_function} then identifies the detection sensitivity with the gauge function of $\Cquantum$:
\begin{proposition}
    For any $C \in \MmR$,
    \begin{equation}
        \sup_{Z \neq 0} \frac{\Tr[Z^\top C]}{\phiquantum (Z)}
        = \gammaquantum(C).
        \label{eq:state-dependent detection sensitivity_beyond_quantum}
    \end{equation}
\end{proposition}
The same reasoning as in Remark~\ref{rem:max_attained} applies: when $\gammaquantum(C) < +\infty$ (equivalently, $\ran C \subset \ran A$ and $\ran C^\top \subset \ran B$), the ratio $\Tr[Z^\top C]/\phiquantum(Z)$ is positively homogeneous of degree zero in $Z$, and $\phiquantum$ is continuous and positive on the unit sphere, so the supremum is attained and $\sup$ may be replaced by $\max$ in~\eqref{eq:state-dependent detection sensitivity_beyond_quantum}.

The gauge function $\gammaquantum(C)$ also admits an interpretation as a noise-robustness measure, by the same reasoning as in Sec.~\ref{sec:detection_of_entanglement}.
By definition, $\gammaquantum(C) = \inf\{t > 0 : C/t \in \Cquantum\}$, so $C/\gammaquantum(C)$ lies on the boundary of $\Cquantum$.
Since mixing with the completely mixed state rescales the correlations as $(1-p)\,C$, the beyond-quantum character of $C$ remains detectable as long as $p < 1 - 1/\gammaquantum(C)$.
The larger $\gammaquantum(C)$, the more depolarizing noise $C$ can withstand before its beyond-quantum character becomes undetectable.

\subsection{Fundamental limit of detection sensitivity}

When arbitrary generalized states may be prepared, the fundamental limit of detection sensitivity for beyond-quantum correlations under fixed measurement settings $A$ and $B$ is obtained by maximizing $\gammaquantum$ over all correlations in $\Cmax$:
\begin{equation}
    \sup_{C \in \Cmax} \gammaquantum(C)
    = \sup_{C \in \Cmax} \max_{Z \neq 0} \frac{\Tr[Z^\top C]}{\phiquantum(Z)}
    = \max_{Z \neq 0} \frac{\phimax(Z)}{\phiquantum(Z)}.
\end{equation}
\begin{theorem}\label{thm:maximum_ratio_of_support_functions_beyond_quantum}
    Let $r := \min\{\rank A, \rank B\}$. Then
    \begin{equation}
        \max_{Z \neq 0} \frac{\phimax (Z)}{\phiquantum (Z)} = \begin{cases}
            3 & \text{if~} r=3, \\
            1 & \text{if~} r\leq 2.
        \end{cases}
    \end{equation}
\end{theorem}
\begin{proof}
    See Appendix~\ref{app:proof_maximum_ratio_bqs}.
\end{proof}
This supremum again coincides with the containment radius, now of $\Cmax$ with respect to $\Cquantum$.
In particular, since $\Cquantum$ is closed, $R(\Cmax, \Cquantum) = 3$ when $r = 3$ means that $\Cmax \subset 3\, \Cquantum$ and that this scaling is optimal.
The case $r \leq 2$ gives $R(\Cmax, \Cquantum) = 1$, so $\Cmax = \Cquantum$: beyond-quantum correlations cannot arise when either party has rank-deficient measurement settings. 
This is consistent with the known result that beyond-quantum ES states are undetectable in the $(2,2,2)$ scenario \cite{Banik2013degree,Stevens2014steering,barnum2013ensemble}.

The following theorem characterizes the correlations $C \in \Cmax$ that achieve the maximum.
\begin{theorem}\label{thm:maximizer_gammaquantum}
    Let $r := \min\{\rank A, \rank B\}$. The following correlations $C \in \Cmax$ achieve $\sup_{C \in \Cmax} \gammaquantum(C)$:
    \begin{enumerate}
        \item When $r = 3$: $C = A Q B^\top$ for any $Q \in \mrm{SO}(3)$.
        \item When $r \leq 2$:
        \begin{itemize}
            \item If $\rank A \geq 2$ and $\rank B \geq 2$: $C = A Q B^\top$ for any $Q \in \mrm{O}(3)$.
            \item If $\rank A = 1$ or $\rank B = 1$: $C = A Q B^\top$ for $Q \in \mrm{O}(3)$ satisfying
            \begin{equation}\label{eq:maximizer_condition_r_leq_2}
                \ran A^\top \cap Q(\ran B^\top) \neq \{0\}.
            \end{equation}
        \end{itemize}
    \end{enumerate}
\end{theorem}
\begin{proof}
    See Appendix~\ref{app:proof_maximizer_gammaquantum}.
\end{proof}

\begin{remark}\label{rem:maximizer_gammaquantum_physical}
    By Remark~\ref{rem:extreme_point_states}, the extreme points of $\Cmax$ are $AQB^\top$ with $Q \in \mrm{O}(3)$: when $Q \in \mrm{SO}^{-}(3)$, the corresponding state is a maximally entangled (quantum) state, while when $Q \in \mrm{SO}(3)$, it belongs to the maximal entanglement structure $\Smax$ but is not a valid quantum state in general.
    By Proposition~\ref{prop:rigidity}, these states are the \emph{only} states in $\Squantum$ and $\Smax$, respectively, that produce the corresponding extremal correlations: no block-positive state with $T \in \mrm{O}(3)$ can have nonzero local Bloch vectors.
    When $r = 3$, Theorem~\ref{thm:maximizer_gammaquantum} shows that the beyond-quantum detection sensitivity $\gammaquantum$ is maximized precisely by the latter class: among all extreme points of $\Cmax$, only those with $Q \in \mrm{SO}(3)$ achieve the maximum.

    The coefficient matrix $Z$ that achieves the maximum ratio $\phimax(Z)/\phiquantum(Z)$ is determined by such a maximizing correlation $C$ via the correspondence in Remark~\ref{rem:gauge_qm_optimizer}.
    When $r = 3$, the maximizer $C_{\star} = A Q B^\top$ with $Q \in \mrm{SO}(3)$ gives $A^{+} C_{\star}\, (B^\top)^{+} = Q$, whose singular value decomposition is $Q = U I_3 V^\top$ with $\det(U V^\top) = +1$.
    By Remark~\ref{rem:gauge_qm_optimizer}, the corresponding optimal coefficient matrix is $Z_{\star} = (A^\top)^{+}\, U V^\top\, B^{+}$, for which $A^\top Z_{\star} B = Q$ has singular values $s_1 = s_2 = s_3 = 1$ and $\det(A^\top Z_{\star} B) = +1$, yielding $\phimax(Z_{\star})/\phiquantum(Z_{\star}) = (1 + 1 + 1)/(1 + 1 - 1) = 3$.
    As in Remark~\ref{rem:maximizer_gammasep_physical}, this simplifies to $Z_{\star} = G_A^{-1}\, C_{\star} G_B^{-1}$.
\end{remark}

\subsection{Noise tolerance: beyond-quantum states}\label{subsec:noise_tolerance_BQS}

We apply the preceding results to the detection of beyond-quantum states under depolarizing noise, in direct analogy with Sec.~\ref{subsec:noise_tolerance_Werner}.
Consider the state
\begin{equation}\label{eq:rho_max}
    \hat{\rho}_{\mrm{max}} :=
    \begin{pmatrix}
        1/2 & 0 & 0 & 0 \\
        0 & 0 & 1/2 & 0 \\
        0 & 1/2 & 0 & 0 \\
        0 & 0 & 0 & 1/2
    \end{pmatrix},
\end{equation}
which has eigenvalues $\{-1/2, 1/2, 1/2, 1/2\}$ and hence is not positive semidefinite---it belongs to $\Smax \setminus \Squantum$.
The noisy state
\begin{equation}\label{eq:tau_p}
    \hat{\tau}_p := (1-p)\, \hat{\rho}_{\mrm{max}} + p\,\frac{\hat{I}_4}{4}
\end{equation}
has eigenvalues $\{(-2+3p)/4,\, (2-p)/4,\, (2-p)/4,\, (2-p)/4\}$.
The state $\hat{\tau}_p$ is a valid quantum state if and only if $(-2+3p)/4 \geq 0$, i.e., $p \geq 2/3$.
Therefore, $\hat{\tau}_p$ is beyond-quantum for $p < 2/3$ and quantum for $p \geq 2/3$.

We consider the Pauli measurement settings $A = B = I_3$, i.e., $\hat{A}_i = \hat{B}_i = \hat{\sigma}_i$ ($i=1,2,3$), so $r = 3$.
The state $\hat{\rho}_{\mrm{max}}$ produces the correlation matrix
\begin{equation}
    C_{\mrm{max}} := \Phi_{AB}(\hat{\rho}_{\mrm{max}}) = A T_{\mrm{max}} B^\top = T_{\mrm{max}},
\end{equation}
where $T_{\mrm{max}} = I_3 \in \mrm{SO}(3)$.
Since $T_{\mrm{max}} \in \mrm{SO}(3)$, the state $\hat{\rho}_{\mrm{max}}$ realizes an extreme point of $\Cmax$ (Theorem~\ref{thm:maximizer_gammaquantum}), and
\begin{equation}
    \gammaquantum(C_{\mrm{max}}) = 3.
\end{equation}
The noisy state $\hat{\tau}_p$ produces the correlation $(1-p)\,C_{\mrm{max}}$, and the beyond-quantum character is detected if and only if $\gammaquantum\bigl((1-p)\,C_{\mrm{max}}\bigr) = (1-p) \cdot 3 > 1$, i.e.,
\begin{equation}\label{eq:p_crit_BQS}
    p < p_{\mrm{crit}}^{\mrm{bqs}} := 1 - \frac{1}{\gammaquantum(C_{\mrm{max}})} = \frac{2}{3}.
\end{equation}
Since $\hat{\tau}_p$ is beyond-quantum precisely for $p < 2/3$, the gauge-function criterion $p_{\mrm{crit}}^{\mrm{bqs}} = 2/3$ achieves necessary and sufficient detection of the beyond-quantum character for this family.
The optimal coefficient matrix is $Z_{\star} = G_A^{-1}\, C_{\mrm{max}}\, G_B^{-1} = C_{\mrm{max}} = I_3$.

\section{Discussion}\label{sec:discussion}

\subsection{Basis-independent expressions for the support and gauge functions}

The support and gauge functions derived in Sections~\ref{sec:support_functions} and~\ref{sec:gauge_functions} are expressed in terms of the matrices $A$ and $B$, which depend on the choice of basis for the operator space.
However, the correlation sets $\Csep$, $\Cquantum$, and $\Cmax$ are determined solely by the observables $\hat{A}_i$ and $\hat{B}_j$, so their support and gauge functions must be basis-independent.
We now verify this explicitly by rewriting all building blocks in terms of intrinsic (basis-independent) quantities.

\subsubsection{Singular values of $A^\top Z B$ and $A^{+}C(B^\top)^{+}$}

Define the Gram matrices $G_A := A A^\top$ and $G_B := B B^\top \in \MmR$, whose entries
\begin{align}
    (G_A)_{ij} &= \bs{a}_i \cdot \bs{a}_j = \frac{1}{2}\Tr[\hat{A}_i \hat{A}_j], \\
    (G_B)_{ij} &= \bs{b}_i \cdot \bs{b}_j = \frac{1}{2}\Tr[\hat{B}_i \hat{B}_j]
\end{align}
are manifestly basis-independent.
Let $X = U_X \Sigma_X V_X^\top$ denote the reduced singular value decomposition of $X = A, B$, so that $G_X^{1/2} = U_X \Sigma_X U_X^\top$.
Then
\begin{align}
    A^\top Z B &= V_A\,(\Sigma_A U_A^\top Z\, U_B \Sigma_B)\,V_B^\top, \\
    G_A^{1/2}\, Z\, G_B^{1/2} &= U_A\,(\Sigma_A U_A^\top Z\, U_B \Sigma_B)\,U_B^\top.
\end{align}
Both matrices share the common core $\Sigma_A U_A^\top Z\, U_B \Sigma_B$ and differ only by left and right isometries ($V_A, V_B$ and $U_A, U_B$, respectively).
Therefore their nonzero singular values coincide, and for any unitarily invariant norm,
\begin{equation}\label{eq:norm_equivalence_support}
    \norm{A^\top Z B}_{p} = \norm{G_A^{1/2}\, Z\, G_B^{1/2}}_{p}, \quad p = \infty,\; 1.
\end{equation}
This establishes the basis-independence of the support functions $\phisep(Z) = \|A^\top Z B\|_\infty$ and $\phimax(Z) = \|A^\top Z B\|_1$.

By the same argument applied to $A^{+}C(B^\top)^{+}$ and $(G_A^{+})^{1/2}\, C\, (G_B^{+})^{1/2}$---using the reduced SVDs of $A^+$ and $(B^\top)^+$---the nonzero singular values of these two matrices also coincide:
\begin{equation}\label{eq:norm_equivalence_gauge}
    \norm{A^{+}C(B^\top)^{+}}_{p} = \norm{(G_A^{+})^{1/2}\, C\, (G_B^{+})^{1/2}}_{p}, \quad p = \infty,\; 1.
\end{equation}
This establishes the basis-independence of the gauge functions $\gammasep(C) = \|A^{+}C(B^\top)^{+}\|_1$ and $\gammamax(C) = \|A^{+}C(B^\top)^{+}\|_\infty$ (in the finite case).

\subsubsection{Determinant of $A^\top Z B$ and $A^{+}C(B^\top)^{+}$}

The support function $\phiquantum$ and the gauge function $\gammaquantum$ involve the sign of $\det(A^\top Z B)$ and $\det(A^{+}C(B^\top)^{+})$, respectively.
To verify their basis-independence, we use the identity (see Lemma~\ref{lem:det_intrinsic} in Appendix~\ref{app:det_identity})
\begin{equation}
    \det(A^\top Z B) = \frac{1}{6}\sum_{i,j,k=1}^{m}\sum_{l,m,n=1}^{m}
      T^A_{ijk}\, z_{il}\, z_{jm}\, z_{kn}\, T^B_{lmn},
    \label{eq:det_intrinsic}
\end{equation}
where
\begin{align}
    T^A_{ijk} &:= (\bs{a}_i \times \bs{a}_j) \cdot \bs{a}_k, \\
    T^B_{ijk} &:= (\bs{b}_i \times \bs{b}_j) \cdot \bs{b}_k.
\end{align}
The tensors $T^A$ and $T^B$ admit basis-independent expressions:
\begin{align}
    T^A_{ijk} &= \frac{1}{4\im}\Tr\!\left([\hat{A}_i, \hat{A}_j]\, \hat{A}_k\right), \\
    T^B_{ijk} &= \frac{1}{4\im}\Tr\!\left([\hat{B}_i, \hat{B}_j]\, \hat{B}_k\right),
\end{align}
since $[\hat{\sigma}_a, \hat{\sigma}_b] = 2\im \sum_c \varepsilon_{abc}\, \hat{\sigma}_c$ and $\Tr[\hat{\sigma}_c \hat{\sigma}_d] = 2\delta_{cd}$.
Therefore $\det(A^\top Z B)$ is basis-independent, and so is the full support function $\phiquantum$.

For the gauge function $\gammaquantum$, it suffices to verify the basis-independence of $\op{sgn}(\det(A^{+}C(B^\top)^{+}))$ when $\rank A = \rank B = 3$.
In this case, $A^{+} = (A^\top A)^{-1} A^\top$ and $(B^\top)^{+} = B(B^\top B)^{-1}$, so
\begin{equation}
    \det(A^{+}C(B^\top)^{+}) = \frac{\det(A^\top C B)}{\det(A^\top A)\,\det(B^\top B)}.
\end{equation}
The numerator $\det(A^\top C B)$ has the same intrinsic expansion~\eqref{eq:det_intrinsic} with $z_{il}$ replaced by $c_{il}$, and the denominators $\det(A^\top A) = \det G_A$ and $\det(B^\top B) = \det G_B$ are positive.
Therefore
\begin{equation}
    \op{sgn}\!\left(\det(A^{+}C(B^\top)^{+})\right) = \op{sgn}\!\left(\sum_{i,j,k}\sum_{l,m,n} T^A_{ijk}\, c_{il}\, c_{jm}\, c_{kn}\, T^B_{lmn}\right),
\end{equation}
which is manifestly basis-independent.
This confirms the basis-independence of the gauge function $\gammaquantum$.

\subsection{Duality structure of support and gauge functions}\label{sec:duality_structure}

The results of Sections~\ref{sec:support_functions} and~\ref{sec:gauge_functions} exhibit a clear duality pattern.
The support and gauge functions of each correlation set take the form of a pair of dual (asymmetric) norms evaluated on the matrices $A^\top Z B$ and $A^{+}C(B^\top)^{+}$, respectively.
The complete correspondence (for the finite cases) is summarized as follows:
\begin{center}
\begin{tabular}{lcc}
\hline\hline
 & Support function & Gauge function \\
\hline
$\Csep$ & $\|A^\top Z B\|_{\infty}$ & $\|A^{+}C(B^\top)^{+}\|_{1}$ \\[2pt]
$\Cquantum$ & $\|A^\top Z B\|_{-}$ & $\|A^{+}C(B^\top)^{+}\|_{+}$ \\[2pt]
$\Cmax$ & $\|A^\top Z B\|_{1}$ & $\|A^{+}C(B^\top)^{+}\|_{\infty}$ \\[2pt]
\hline\hline
\end{tabular}
\end{center}
where we define
\begin{equation}\label{eq:asymmetric_norm_pm}
    \|X\|_{\pm} := s_1 + s_2 \pm s_3\, \op{sgn}(\det X)
\end{equation}
for any $3 \times 3$ real matrix $X$ with singular values $s_1 \geq s_2 \geq s_3 \geq 0$.
For $\Csep$ and $\Cmax$, the support and gauge functions form the standard dual pair of the trace norm and operator norm:
\begin{equation}
    \|\cdot\|_{1} \overset{\text{dual}}{\longleftrightarrow} \|\cdot\|_{\infty}.
\end{equation}
For $\Cquantum$, the duality involves the asymmetric norms $\|\cdot\|_{+}$ and $\|\cdot\|_{-}$.
These are not standard norms: they lack absolute homogeneity since $\det(-X) = -\det X$ for $3 \times 3$ matrices, so $\|-X\|_{\pm} = \|X\|_{\mp} \neq \|X\|_{\pm}$ in general.
Instead, $\|\cdot\|_{+}$ and $\|\cdot\|_{-}$ form a dual pair of asymmetric norms in the sense that
\begin{equation}\label{eq:asymmetric_norm_duality}
    \|X\|_{+} = \sup_{Y \neq 0} \frac{\Tr[X^\top Y]}{\|Y\|_{-}}, \qquad
    \|X\|_{-} = \sup_{Y \neq 0} \frac{\Tr[X^\top Y]}{\|Y\|_{+}}.
\end{equation}
This is proved as Lemma~\ref{lem:asymmetric_norm_duality} in the appendix by identifying $\{Y : \|Y\|_{-} \leq 1\} = \op{conv}(\mrm{SO}(3))$ and computing the support function of this convex body.

The above duality structure is inherited from that of the support and gauge functions of the corresponding state spaces, as we now explain.
Define $\tilde{\mcal{S}}_{\text{model}} := \mcal{S}_{\text{model}} - \hat{I}_4/4$ for $\text{model} = \text{sep}, \,\text{qm}, \, \text{max}$.
Each $\tilde{\mcal{S}}_{\text{model}}$ is a closed convex subset of $\mcal{B}^{\mrm{S}}(\CtwoCtwo)$ containing the origin.
Therefore, $\tilde{\mcal{S}}_{\text{model}}$ is the unit ball of its gauge function:
\begin{equation}
    \tilde{\mcal{S}}_{\text{model}}
    = \left\{
        \hat{V} \in \mcal{B}^{\mrm{S}}(\CtwoCtwo) \,:\,
        \gamma_{\tilde{\mcal{S}}_{\text{model}} } (\hat{V}) \leq 1
    \right\},
\end{equation}
and the support function $\phi_{\tilde{\mcal{S}}_{\text{model}} }$ is the dual of the gauge function:
\begin{equation}
    \phi_{\tilde{\mcal{S}}_{\text{model}} }(\hat{W})
    = \sup_{\hat{V} \in \mcal{B}^{\mrm{S}}(\CtwoCtwo)} \left\{\Tr\left[\hat{W}\,\hat{V}\right] \,:\, \gamma_{\tilde{\mcal{S}}_{\text{model}} }(\hat{V}) \leq 1\right\}.
\end{equation}
Note that these support and gauge functions are not seminorms, since the sets $\tilde{\mcal{S}}_{\text{model}}$ are not symmetric about the origin.

The support function of $\Cmodel$ is the composition of $\phi_{\tilde{\mcal{S}}_{\text{model}}}$ with the Bell operator map $\hat{S}_{AB}: Z \mapsto \sum_{i,j} z_{ij}\, \hat{A}_i \otimes \hat{B}_j$.
Since $\Tr[\hat{S}_{AB}(Z)] = 0$, the shift by $\hat{I}_4/4$ drops out:
\begin{align*}
    \phi_{\text{model}}^{AB}(Z) &= \sup_{\hat{\rho} \in \mcal{S}_{\text{model}}} \Tr\left[\hat{S}_{AB}(Z) \, \hat{\rho}\right] \\
    &= \sup_{\hat{V} \,:\, \gamma_{\tilde{\mcal{S}}_{\text{model}}}(\hat{V}) \leq 1} \Tr\left[\hat{S}_{AB}(Z) \, \hat{V}\right] \\
    &= \phi_{\tilde{\mcal{S}}_{\text{model}}} (\hat{S}_{AB}(Z)).
    \numberthis \label{eq:support_pullback}
\end{align*}
Therefore, $\phi_{\text{model}}^{AB}$ is the pullback of $\phi_{\tilde{\mcal{S}}_{\text{model}}}$ through the map~$\hat{S}_{AB}$.

For the gauge function, define $\Psi_{AB} (\hat{V}) := \Phi_{AB}(\hat{I}_4/4 + \hat{V})$.
Since $\Phi_{AB}(\hat{I}_4/4) = 0$, the map $\Psi_{AB}$ is the restriction of $\Phi_{AB}$ to the traceless part, and $\Psi_{AB}$ is linear.
Then
\begin{align*}
    C \in \Cmodel & \iff \exists \, \hat{V} \in \mcal{B}^{\mrm{S}}(\CtwoCtwo) \;\text{s.t.}\; C = \Psi_{AB}(\hat{V}),\; \gamma_{\tilde{\mcal{S}}_{\text{model}}} (\hat{V}) \leq 1 \\
    & \iff \inf_{\hat{V}} \left\{
        \gamma_{\tilde{\mcal{S}}_{\text{model}}} (\hat{V}) \,:\, C = \Psi_{AB}(\hat{V})
    \right\} \leq 1 \\
    & \iff \gamma_{\tilde{\mcal{S}}_{\text{model}}} \left( \Psi_{AB}^{\spi} (C) \right) \leq 1,
\end{align*}
where $\Psi_{AB}^{\spi}$ denotes the standard partial inverse of $\Psi_{AB}$~\cite{lee2022universal}, i.e., the minimum-norm preimage under $\Psi_{AB}$.
Therefore
\begin{equation}\label{eq:gauge_pushforward}
    \gamma_{\text{model}}^{AB}(C) = \gamma_{\tilde{\mcal{S}}_{\text{model}}} \left(\Psi_{AB}^{\spi}(C)\right),
\end{equation}
and the gauge function $\gamma_{\text{model}}^{AB}$ is the pushforward of $\gamma_{\tilde{\mcal{S}}_{\text{model}}}$ through $\Psi_{AB}$.

Equations~\eqref{eq:support_pullback} and~\eqref{eq:gauge_pushforward} reveal that the duality between the support and gauge functions at the level of correlation sets is inherited from the duality between $\phi_{\tilde{\mcal{S}}_{\text{model}}}$ and $\gamma_{\tilde{\mcal{S}}_{\text{model}}}$ at the level of state spaces.
The overall structure is depicted in Fig.~\ref{fig:duality_diagram}.

\begin{figure}[t]
\centering
\includegraphics[width=150mm]{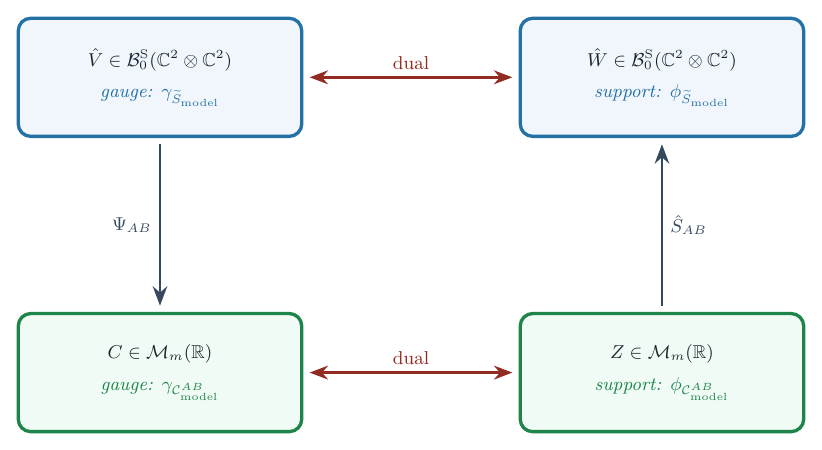}

\caption{Duality structure of the support and gauge functions.
The upper level represents the traceless subspace $\mcal{B}^{\mrm{S}}_0(\CtwoCtwo) := \{\hat{X} \in \mcal{B}^{\mrm{S}}(\CtwoCtwo) : \Tr[\hat{X}] = 0\}$ of the operator space, where the gauge function $\gamma_{\tilde{\mcal{S}}_{\text{model}}}$ and the support function $\phi_{\tilde{\mcal{S}}_{\text{model}}}$ form a dual pair.
The lower level represents the correlation matrix space $\MmR$.
The Bell operator map $\hat{S}_{AB}$ pulls back the support function from $\hat{W}$ to $Z$ (Eq.~\eqref{eq:support_pullback}), while the map $\Psi_{AB}$ pushes forward the gauge function from $\hat{V}$ to $C$ (Eq.~\eqref{eq:gauge_pushforward}).}
\label{fig:duality_diagram}
\end{figure}

A natural direction for future work is to extend this framework to state spaces constrained by Schmidt rank or Schmidt number, where the relevant norms are the $k$-norms studied in Ref.~\cite{JohnstonKribs2013}.

\subsection{Symmetry of correlation sets}

At the state-space level, all three shifted state spaces $\tilde{\mcal{S}}_{\text{model}}$ are asymmetric about the origin, so the gauge and support functions $\gamma_{\tilde{\mcal{S}}_{\text{model}}}$ and $\phi_{\tilde{\mcal{S}}_{\text{model}}}$ are not seminorms.
The correlation map, however, treats the three models differently.

The correlation sets $\Csep$ and $\Cmax$ are symmetric: $C \in \Csep \iff -C \in \Csep$, and likewise for $\Cmax$.
Consequently, their gauge and support functions reduce to genuine norms (the trace norm and operator norm), despite the asymmetry of the underlying state spaces.
To see why, consider the map
\begin{equation}\label{eq:symmetrization_map}
    \Theta : \hat{\rho} \;\mapsto\; (\hat{I}_2 \otimes \hat{\sigma}_2)\, (\mrm{id} \otimes \top)(\hat{\rho})\, (\hat{I}_2 \otimes \hat{\sigma}_2),
\end{equation}
which is the composition of partial transposition on Bob's subsystem with conjugation by the local unitary $\hat{I}_2 \otimes \hat{\sigma}_2$.
Since $\hat{\sigma}_2 \hat{\sigma}_k \hat{\sigma}_2 = -\hat{\sigma}_k^\top$ for $k = 1,2,3$, this map reverses the sign of all correlations: $\Phi^{AB}(\Theta(\hat{\rho})) = -\Phi^{AB}(\hat{\rho})$.
Moreover, $\Theta$ preserves both $\Ssep$ and $\Smax$.
Let $\Gamma := \mrm{id} \otimes \top$ denote partial transposition on Bob's subsystem.
Since $\Gamma$ maps every product state $\hat{\rho}_A \otimes \hat{\rho}_B$ to $\hat{\rho}_A \otimes \hat{\rho}_B^\top$, which is again a product state, we have $\Gamma(\op{SEP}) = \op{SEP}$.
The invariance of $\op{SEP}^*$ follows from the general theory of dual cones.
The map $\Gamma$ is self-adjoint with respect to the Hilbert--Schmidt inner product,
\begin{equation}\label{eq:partial_transp_self_adjoint}
    \Tr\bigl[M\,\Gamma(\hat{\rho})\bigr] = \Tr\bigl[\Gamma(M)\,\hat{\rho}\bigr]
    \quad \text{for all } M \in \op{SEP},\hat{\rho} \in \op{SEP}^*.
\end{equation}
and satisfies $\Gamma^2 = \mrm{id}$, so $\Tr[M\,\hat{\rho}] = \Tr[\Gamma(M)\,\Gamma(\hat{\rho})]$.
Therefore,
\begin{align}
    \hat{\rho} \in \op{SEP}^*
    &\iff \forall M \in \op{SEP},\; \Tr[M\,\hat{\rho}] \geq 0 \notag\\
    &\iff \forall M \in \op{SEP},\; \Tr[\Gamma(M)\,\Gamma(\hat{\rho})] \geq 0 \notag\\
    &\iff \forall M' \in \Gamma(\op{SEP}),\; \Tr[M'\,\Gamma(\hat{\rho})] \geq 0 \notag\\
    &\iff \Gamma(\hat{\rho}) \in \op{SEP}^*, \label{eq:sep_star_invariance}
\end{align}
where the last step uses $\Gamma(\op{SEP}) = \op{SEP}$.
Hence $\Gamma(\op{SEP}^*) = \op{SEP}^*$ (and consequently $\Gamma(\Smax) = \Smax$).
Invariance under $\Theta$ follows by the same reasoning: conjugation by the local unitary $\hat{I}_2 \otimes \hat{\sigma}_2$ is a bijection on $\mcal{B}^{\mrm{S}}(\CtwoCtwo)$ that preserves both $\op{SEP}$ and $\op{SEP}^*$, so the chain~\eqref{eq:sep_star_invariance} applies to $\Theta$ as well.
Therefore $C \in \Csep$ implies $-C = \Phi^{AB}(\Theta(\hat{\rho})) \in \Csep$, and similarly for $\Cmax$.

In contrast, the quantum state space $\Squantum$ is \emph{not} invariant under $\Theta$---partial transposition can map an entangled state outside $\Squantum$---so the above argument does not apply to $\Cquantum$.
Indeed, when $\rank A, \rank B \geq 3$, the set $\Cquantum$ is not symmetric, and the asymmetry of $\Squantum$ about the maximally mixed state survives in the correlation space as the asymmetric norms $\|\cdot\|_{\pm}$.
However, when $\rank A \leq 2$ or $\rank B \leq 2$, the set $\Cquantum$ is symmetric.
In this case, the observables of the party with rank-deficient settings span at most a two-dimensional subspace of $\mbb{R}^3$, so there exists a unit vector $\bs{n}$ orthogonal to all of them.
The rotation by $\pi$ about $\bs{n}$ is a local unitary that flips the sign of all observables in this subspace while preserving $\Squantum$, thereby providing the required symmetry.

\subsection{Role of $m$ and $r$ in distinguishing the correlation sets}
The maximum ratio results (Theorems~\ref{thm:maximum_ratio_of_support_functions} and~\ref{thm:maximum_ratio_of_support_functions_beyond_quantum}) reveal that the distinguishability of $\Csep$, $\Cquantum$, and $\Cmax$ is governed by $r := \min\{\rank A, \rank B\}$, which counts the number of independent physical observables per party.

When $r = 1$, each party has only one independent observable, and all three sets coincide: $\Csep = \Cquantum = \Cmax$. When $r = 2$---which includes the CHSH scenario ($m = 2$) as well as $m \geq 3$ with rank-deficient measurement matrices---$\Csep$ and $\Cquantum$ can be distinguished (the maximum ratio of support functions equals~$2$), but $\Cquantum = \Cmax$, so beyond-quantum correlations cannot be detected. Only when $m \geq 3$ and $r = 3$ do all three sets become distinct: $\Csep \subsetneq \Cquantum \subsetneq \Cmax$, with the maximum ratios equal to~$3$. In particular, even for $m \geq 3$, the sets cannot be fully distinguished unless $r = 3$, i.e., unless both parties have three linearly independent measurement directions.

\subsection{Generalizations}\label{subsec:generalizations}
The results of this paper exploit the structure of two-qubit systems in a crucial way.
We discuss how the derivations extend---and where they break down---in multipartite and higher-dimensional settings.

As discussed in the previous subsection, the distinguishability of the three correlation sets is governed by $r = \min\{\rank A, \rank B\}$, the number of linearly independent measurement directions per party.
For two-qubit systems, each party's Bloch vectors lie in $\mbb{R}^3$, so $r \leq 3$.
In particular, once $m > 3$, increasing the number of measurement choices per party does not increase $r$, and the additional observables provide no further discrimination power.
This limitation motivates extending the framework to multipartite systems and higher-dimensional local Hilbert spaces, where larger values of $r$ become accessible.

\subsubsection{Multipartite settings}
Consider the $(n,m,2)$ scenario, in which $n$ parties each perform $m$ dichotomic measurements on a local qubit.
The observables of the $p$-th party ($p = 1,\ldots,n$) are
\begin{equation}
    \hat{A}^{(p)}_i = \bs{a}^{(p)}_i \cdot \hat{\bs{\sigma}}, \quad i \in \{1,\ldots,m\},
\end{equation}
where $\bs{a}^{(p)}_i \in \mbb{R}^3$ are unit vectors.
We collect them into matrices
\begin{equation}
    A^{(p)} := \begin{bmatrix}
        (\bs{a}^{(p)}_1)^\top \\ \vdots \\ (\bs{a}^{(p)}_m)^\top
    \end{bmatrix} \in \mcal{M}_{m,3}(\mbb{R}), \quad
    p \in \{1,\ldots,n\}.
\end{equation}

For the separable correlation set, the extreme points of the fully separable state space are pure product states, so the support function reduces to a supremum over local Bloch vectors (cf.\ the bipartite proof in Appendix~\ref{app:proof_support_Csep}).
Since the correlation produced by a product state $\hat{\rho}_1 \otimes \cdots \otimes \hat{\rho}_n$ factorizes as $c_{i_1 \cdots i_n} = \prod_p (\bs{a}^{(p)}_{i_p} \cdot \bs{r}_p)$, this gives
\begin{align*}
    \phi_{\text{sep}}^{\{A^{(p)}\}}(Z)
    &= \sup_{\|\bs{r}_1\|=\cdots=\|\bs{r}_n\|=1} \sum_{i_1,\ldots,i_n} z_{i_1 \cdots i_n} \prod_{p=1}^{n} \bs{a}^{(p)}_{i_p} \cdot \bs{r}_p \\
    &= \norm{T^{\{A^{(p)}\}}_Z}_{\mrm{inj}},
    \numberthis
\end{align*}
where
\begin{equation}
   \left( T^{\{A^{(p)}\}}_Z \right)_{j_1 \cdots j_n}
   := \sum_{i_1, \ldots, i_n} z_{i_1 \cdots i_n} \prod_{p=1}^{n} a^{(p)}_{i_p j_p},
\end{equation}
and $\|\cdot\|_{\mrm{inj}}$ denotes the injective tensor norm:
\begin{equation}
    \|T\|_{\mrm{inj}}
    := \sup_{\|\bs{r}_1\|=\cdots=\|\bs{r}_n\|=1}
    \sum_{j_1, \ldots, j_n} T_{j_1 \cdots j_n} \prod_{p=1}^{n} (r_p)_{j_p}.
\end{equation}
{Computing the injective norm of a general tensor is NP-hard~\cite{hillar2013most},} so a closed-form expression analogous to the bipartite operator norm $\|A^\top Z B\|_\infty$ is unlikely to exist.

For the quantum correlation set, the support function equals the largest eigenvalue of the Bell operator (cf.\ Appendix~\ref{app:proof_support_Cquantum} for the bipartite case):
\begin{equation}
    \phi_{\text{qm}}^{\{A^{(p)}\}} (Z) = \lambda_{\max} \!\left(\hat{S}_{\{A^{(p)}\}}(Z)\right),
\end{equation}
where
\begin{equation}
    \hat{S}_{\{A^{(p)}\}}(Z) := \sum_{i_1,\ldots,i_n} z_{i_1 \cdots i_n}\, \hat{A}^{(1)}_{i_1} \otimes \cdots \otimes \hat{A}^{(n)}_{i_n}.
\end{equation}
In the bipartite case, the SVD of $A^\top Z B$ reduces the Bell operator to the diagonal form $\sum_k s_k\, \hat{\sigma}_k \otimes \hat{\sigma}_k$, where the three operators $\hat{\sigma}_k \otimes \hat{\sigma}_k$ commute and are simultaneously diagonalizable, yielding an explicit formula for $\lambda_{\max}$ (see Appendix~\ref{app:proof_support_Cquantum}).
In the multipartite case, {the analogous decomposition is the higher-order singular value decomposition (HOSVD)~\cite{Lathauwer2000multilinear}, whose core tensor is not diagonal in general.}
Therefore, an explicit analytic formula for $\lambda_{\max}$ does not follow from the same argument.

For the maximal correlation set, the bipartite derivation relies on the fact that the extreme points of $\Smax$ are precisely the quantum pure states and their partial transposes~\cite[Proposition~12]{arai2024detecting_supplemental}.
This characterization does not extend to the multipartite or higher-dimensional cases, where the extreme-point structure of $\mcal{S}(\op{SEP}^*)$ is more complex.

\subsubsection{Higher-dimensional local Hilbert spaces}

For a bipartite system with local Hilbert space $\mbb{C}^d$, the generalized Bloch vector of a state lies in $\mbb{R}^{d^2-1}$, so the measurement matrices have $d^2-1$ columns and $r$ can be as large as $d^2-1$.
This makes the higher-dimensional setting a natural arena for accessing richer correlation structures.

Consider the bipartite $(2,m,2)$ scenario with local Hilbert space $\mbb{C}^d$ for general $d \geq 2$.
A natural analogue of the Pauli basis is a set of generators of $\mrm{SU}(d)$: the traceless self-adjoint operators $\{\hat{\lambda}_i\}_{i=1}^{d^2-1}$ satisfying $\Tr[\hat{\lambda}_i \hat{\lambda}_j] = \delta_{ij}$.
Every self-adjoint operator on $\mbb{C}^d$ admits the expansion
\begin{equation}
    \hat{X} = \frac{\Tr[\hat{X}]}{d}\,\hat{I}_d
    + \sum_{i=1}^{d^2-1} x_i\, \hat{\lambda}_i,
\end{equation}
where $x_i := \Tr[\hat{X}\,\hat{\lambda}_i]$.
The measurement matrices $A \in \mcal{M}_{m,d^2-1}(\mbb{R})$ and $B \in \mcal{M}_{m,d^2-1}(\mbb{R})$ are defined by $A_{ki} := \Tr[\hat{A}_k\, \hat{\lambda}_i]$ and $B_{ki} := \Tr[\hat{B}_k\, \hat{\lambda}_i]$.
However, extending the support and gauge function formulae to this setting faces two distinct obstacles.

For the separable correlation set, the proof in the qubit case (Appendix~\ref{app:proof_support_Csep}) reduces $\phisep(Z)$ to the supremum of the bilinear form $\bs{r}_A^\top (A^\top Z B)\, \bs{r}_B$ over local Bloch vectors $\bs{r}_A, \bs{r}_B$.
Since the qubit Bloch body $\mrm{Bl}(2)$ is the unit ball, this supremum equals the operator norm $\|A^\top Z B\|_\infty$.
For $d \geq 3$, the same reduction gives
\begin{equation}
    \phisep(Z) = \sup_{\bs{r}_A, \bs{r}_B \in \mrm{Bl}(d)} \bs{r}_A^\top (A^\top Z B)\, \bs{r}_B,
\end{equation}
where $\mrm{Bl}(d) := \{\bs{r} \in \mbb{R}^{d^2-1} : \frac{1}{d}\hat{I}_d + \sum_i r_i \hat{\lambda}_i \geq 0\}$ is the generalized Bloch body.
However, $\mrm{Bl}(d)$ is a strict subset of a ball with a more complex geometry~\cite{kimura2003bloch}, so the supremum does not simplify to a standard matrix norm.

For the quantum correlation set, the support function is still the largest eigenvalue of the Bell operator (by the same reasoning as in the qubit case; see Appendix~\ref{app:proof_support_Cquantum}):
\begin{equation}
    \hat{S}_{AB}(Z) := \sum_{i,j} z_{ij}\, \hat{A}_i \otimes \hat{B}_j
    = \sum_{k,l} (A^\top Z B)_{kl}\, \hat{\lambda}_k \otimes \hat{\lambda}_l.
\end{equation}
Applying the SVD $A^\top Z B = U \Sigma V^\top$ and rotating the operator basis as $\hat{\mu}_k := \sum_l U_{lk} \hat{\lambda}_l$, $\hat{\nu}_k := \sum_l V_{lk} \hat{\lambda}_l$, the Bell operator reduces to $\hat{S}_{AB}(Z) = \sum_k s_k\, \hat{\mu}_k \otimes \hat{\nu}_k$, whose eigenvalues coincide with those of
\begin{equation}
    \sum_{k=1}^{d^2-1} s_k\, \hat{\lambda}_k \otimes \hat{\lambda}_k,
\end{equation}
where $s_1 \geq \cdots \geq s_{d^2-1} \geq 0$ are the singular values of $A^\top Z B$.
For $d = 2$, the three operators $\hat{\sigma}_k \otimes \hat{\sigma}_k$ ($k = 1,2,3$) commute and are simultaneously diagonalizable, which enables the explicit eigenvalue formula.
For $d \geq 3$, the operators $\hat{\lambda}_k \otimes \hat{\lambda}_k$ do not commute in general and lack a common eigenbasis, so deriving an analytic formula for the largest eigenvalue remains an open problem.

\subsection{Duality between the separable and maximal correlation sets}

The separable and maximal state spaces arise from dual cones: $\Ssep = \mcal{S}(\op{SEP})$ and $\Smax = \mcal{S}(\op{SEP}^*)$.
In our two-qubit $(2,m,2)$ results, the support and gauge functions of $\Csep$ and $\Cmax$ exhibit a suggestive pattern reminiscent of this cone-level duality.
As summarized in the table in Sec.~\ref{sec:duality_structure}, the roles of the trace norm and operator norm are exchanged between $\Csep$ and $\Cmax$:
\begin{alignat}{2}
    \phisep(Z) &= \|A^\top Z B\|_\infty, &\qquad \gammasep(C) &= \|A^{+}C(B^\top)^{+}\|_1, \\
    \phimax(Z) &= \|A^\top Z B\|_1, &\qquad \gammamax(C) &= \|A^{+}C(B^\top)^{+}\|_\infty.
\end{alignat}
Specifically, two dual relationships are visible:
the support functions $\phisep = \|\cdot\|_\infty$ and $\phimax = \|\cdot\|_1$ form a dual pair, and the gauge functions $\gammasep = \|\cdot\|_1$ and $\gammamax = \|\cdot\|_\infty$ form another dual pair.
It is tempting to attribute these dualities to the cone-level duality $\op{SEP} \leftrightarrow \op{SEP}^*$ at the state-space level.
However, this connection is not straightforward: the support and gauge functions of the correlation sets are obtained from those of the state spaces through the pullback and pushforward by the correlation map (Eqs.~\eqref{eq:support_pullback} and~\eqref{eq:gauge_pushforward}), and it is a priori unclear whether the cone duality survives this process in a form as simple as a norm duality.

In the two-qubit $(2,m,2)$ scenario studied here, the correlation sets are determined by the singular values of a $3 \times 3$ matrix, which admits a complete characterization by the trace norm and operator norm.
Whether an analogous norm-level duality between the separable and maximal correlation sets persists in more general settings---such as multipartite scenarios, higher-dimensional local Hilbert spaces, or scenarios with more than two outcomes per measurement---remains an open question.

\section{Conclusion}\label{sec:conclusion}

We derived explicit formulas for the support functions and gauge functions of the three device-dependent correlation sets---$\Csep$, $\Cquantum$, and $\Cmax$---in the $(2,m,2)$ Bell scenario for two-qubit systems.
All formulas are expressed as pairs of dual matrix norms applied to common variable transformations: $A^\top Z B$ for the support functions and $A^{+} C (B^\top)^{+}$ for the gauge functions.
For the separable and maximal correlation sets, the relevant norms are the operator norm~$\|\cdot\|_\infty$ and trace norm~$\|\cdot\|_1$, interchanged between the two sets.
For the quantum correlation set, the asymmetric norms $\|\cdot\|_\pm$, defined via $s_1 + s_2 \pm s_3 \operatorname{sgn}(\det)$, emerge naturally; this asymmetry reflects the behavior of the quantum state space under partial transposition.

We applied these functions to the detection of entanglement and beyond-quantum states.
The gauge function $\gammasep$ provides a quantitative robustness measure for entanglement detection: for Werner states, the optimal critical noise fraction is $p_{\mrm{crit}} = 1 - 1/r$, where $r := \min\{\rank A, \rank B\}$.
When $r = 3$, this threshold $p_{\mrm{crit}} = 2/3$ coincides with the PPT threshold; the PPT criterion is necessary and sufficient for separability in two-qubit systems~\cite{horodecki1996separability}. This exceeds the thresholds of device-independent tests such as the CHSH and $I_{3322}$ inequalities.
Analogously, the gauge function $\gammaquantum$ provides a quantitative robustness measure for beyond-quantum state detection: for a family of noisy beyond-quantum states $\hat{\tau}_p = (1-p)\,\hat{\rho}_{\mrm{max}} + p\,\hat{I}_4/4$ with $\hat{\rho}_{\mrm{max}} \in \Smax \setminus \Squantum$, the critical noise fraction is $p_{\mrm{crit}}^{\mrm{bqs}} = 2/3$ when $r = 3$, which coincides with the boundary beyond which $\hat{\tau}_p$ fails to be a valid quantum state.
When $r \leq 2$, $\Cmax = \Cquantum$ and beyond-quantum correlations are undetectable; beyond-quantum states in the maximal entanglement structure are thus detectable via correlations if and only if both parties possess three linearly independent measurement directions.
These detection limits depend only on $r$ and are independent of the specific measurement directions.

Several directions for generalization remain open, including extensions to multipartite systems---where the separable support function reduces to the injective tensor norm, whose computation is NP-hard in general---and to higher-dimensional local Hilbert spaces, where the Bloch body geometry is more complex than the unit ball.
The convex-analytic framework developed here, which provides a complete dual characterization of device-dependent correlation sets via support and gauge functions, provides a starting point for such extensions.

\begin{acknowledgments}

\end{acknowledgments}
\appendix

\section{Proofs of the Support Function Theorems}\label{appsec:proofs_support_functions}

We begin with a standard identity that relates matrix traces to Kronecker products via vectorization.
\begin{lemma}[Vectorization identity]\label{lem:vectorization}
    For a $p \times q$ matrix $M$ and matrices $X$, $Y$ of compatible sizes,
    \begin{equation}
        \op{vec}(XMY) = (Y^\top \otimes X)\, \op{vec}(M),
        \label{eq:vectorization_identity}
    \end{equation}
    where $\op{vec}(M)$ denotes the vector obtained by stacking the columns of $M$.
\end{lemma}
\begin{proof}
    Let $\{\bs{e}_j\}_{j=1}^{q}$ and $\{\bs{f}_k\}_{k=1}^{r}$ denote the standard bases of $\mbb{R}^q$ and $\mbb{R}^r$, respectively, so that $M = \sum_{j,k} m_{jk}\, \bs{e}_j \bs{f}_k^\top$.
    Then
    \begin{equation}
        \op{vec}(XMY)
        = \sum_{j,k} m_{jk}\, \op{vec}(X \bs{e}_j \bs{f}_k^\top Y).
    \end{equation}
    Since $\op{vec}(\bs{u}\, \bs{v}^\top) = \bs{v} \otimes \bs{u}$ for any column vectors $\bs{u}$ and $\bs{v}$, we have
    \begin{equation}
        \op{vec}(X \bs{e}_j\, (Y^\top \bs{f}_k)^\top)
        = Y^\top \bs{f}_k \otimes X \bs{e}_j
        = (Y^\top \otimes X)(\bs{f}_k \otimes \bs{e}_j).
    \end{equation}
    Substituting this back and using $\op{vec}(M) = \sum_{j,k} m_{jk}\, (\bs{f}_k \otimes \bs{e}_j)$ yields \eqref{eq:vectorization_identity}.
\end{proof}

\subsection{Proof of Theorem~\ref{thm:support_function_of_Csep}: Support function of $\Csep$}
\label{app:proof_support_Csep}
\begin{proof}
\noindent\textbf{General case.}
For a product state $\hat{\rho}_A \otimes \hat{\rho}_B$ with Bloch vectors $\bs{r}_A, \bs{r}_B \in \mbb{R}^3$ ($\|\bs{r}_A\|, \|\bs{r}_B\| \leq 1$), the correlation matrix has entries
\begin{equation}
    c_{ij} = \Tr[\hat{\rho}_A\, \hat{A}_i]\, \Tr[\hat{\rho}_B\, \hat{B}_j]
    = (\bs{a}_i \cdot \bs{r}_A)(\bs{b}_j \cdot \bs{r}_B)
    = (A \bs{r}_A)_i\, (B \bs{r}_B)_j,
\end{equation}
so $C = (A \bs{r}_A)(B \bs{r}_B)^\top$.
Since the extreme points of $\Ssep$ are pure product states ($\|\bs{r}_A\| = \|\bs{r}_B\| = 1$) and the support function is the supremum of a linear functional over a compact convex set, we obtain
\begin{align*}
    \phisep(Z)
    &= \sup_{\hat{\rho} \in \Ssep} \Tr[Z^\top \Phi_{AB}(\hat{\rho})]\\
    &= \sup_{\|\bs{r}_A\| = \|\bs{r}_B\| = 1} \Tr\!\left[Z^\top (A \bs{r}_A)(B \bs{r}_B)^\top\right] \\
    &= \sup_{\|\bs{r}_A\| = \|\bs{r}_B\| = 1} \bs{r}_A^\top (A^\top Z B)\, \bs{r}_B \numberthis \label{eq:app_phisep_sup}\\
    &= \opnorm{A^\top Z B} \\
    &= s_1, \numberthis
    \label{eq:app_phisep_general}
\end{align*}
where the last equality follows from the variational characterization of the operator norm (e.g., Theorem 3.4.1 in \cite{Horn1994topics}).

\noindent\textbf{Specialization to $m=2$.}
When $m=2$, the $3 \times 3$ matrix $A^\top Z B$ has rank at most~$2$, so its singular values are $s_1 \geq s_2 \geq 0$ (with $s_3 = 0$).
Since $\bs{a}_i$ and $\bs{b}_j$ are unit vectors, the Gram matrices satisfy
\begin{equation}
    A A^\top = \begin{pmatrix} 1 & \bs{a}_1 \cdot \bs{a}_2 \\ \bs{a}_1 \cdot \bs{a}_2 & 1 \end{pmatrix} = G_\alpha,
    \qquad
    B B^\top = G_\beta.
\end{equation}
We define the $2 \times 2$ positive semidefinite matrix
\begin{equation}
    P := G_\alpha\, Z\, G_\beta\, Z^\top.
\end{equation}
The squared singular values $s_1^2, s_2^2$ are the nonzero eigenvalues of the $3\times 3$ matrix $(A^\top Z B)^\top (A^\top Z B) = B^\top Z^\top G_\alpha\, Z\, B$.
Since $XY$ and $YX$ share the same nonzero eigenvalues (with multiplicities) for any matrices $X$, $Y$ of compatible sizes,
taking $X = B^\top$ and $Y = Z^\top G_\alpha\, Z\, B$ shows that these coincide with the eigenvalues of $Z^\top G_\alpha\, Z\, G_\beta$.
Applying the same property once more with $X = Z^\top$ and $Y = G_\alpha\, Z\, G_\beta$, we conclude that $s_1^2, s_2^2$ are precisely the eigenvalues of $P = G_\alpha\, Z\, G_\beta\, Z^\top$.
The trace and determinant of $P$ give
\begin{equation}
    s_1^2 + s_2^2 = \Tr P, \qquad
    s_1^2\, s_2^2 = \det P = \sin^2\!\alpha\, \sin^2\!\beta\, (\det Z)^2.
    \label{eq:app_sep_trace_det}
\end{equation}
The largest eigenvalue of a $2 \times 2$ matrix with trace $T$ and determinant $D$ is $(T + \sqrt{T^2 - 4D}\,)/2$, so
\begin{equation}
    s_1^2 = \frac{\Tr P + \sqrt{(\Tr P)^2 - 4 \det P}}{2}.
    \label{eq:app_sep_s1_eigenvalue}
\end{equation}

We claim that $\sqrt{(\Tr P)^2 - 4 \det P} = |\Tr[G_\alpha\, Z\, H_\beta\, Z^\top]|$.
To verify this, define the real symmetric matrix $R := Z^\top G_\alpha\, Z$.
Then $\Tr P = \Tr[G_\beta\, R]$ and $\Tr[G_\alpha\, Z\, H_\beta\, Z^\top] = \Tr[H_\beta\, R]$.
Writing $R = \bigl(\begin{smallmatrix} a & c \\ c & b \end{smallmatrix}\bigr)$, we compute
\begin{align}
    \Tr[G_\beta\, R] &= a + b + 2c \cos\beta, \\
    \Tr[H_\beta\, R] &= a\, \ep{-\im\beta} + b\, \ep{\im\beta} + 2c.
\end{align}
A direct calculation gives
\begin{equation}
    |\Tr[H_\beta\, R]|^2
    = \bigl((a{+}b)\cos\beta + 2c\bigr)^2 + (b{-}a)^2 \sin^2\!\beta
    = (\Tr[G_\beta\, R])^2 - 4 \sin^2\!\beta\, \det R.
    \label{eq:app_sep_H_identity}
\end{equation}
Since $\det R = \sin^2\!\alpha\, (\det Z)^2$, we have $4 \sin^2\!\beta\, \det R = 4 \det P$, which proves the claim.
Substituting into~\eqref{eq:app_sep_s1_eigenvalue} yields the matrix form~\eqref{eq:support_function_of_Csep_matrix_form}:
\begin{equation}
    \phisep(Z) = \sqrt{\frac{\Tr[G_\alpha\, Z\, G_\beta\, Z^\top] + |\Tr[G_\alpha\, Z\, H_\beta\, Z^\top]|}{2}}.
\end{equation}

Finally, Lemma~\ref{lem:vectorization} gives
\begin{align}
    \Tr[G_\alpha\, Z\, G_\beta\, Z^\top] &= \bs{z}^\top (G_\beta \otimes G_\alpha)\, \bs{z} = \bs{z}^\top K_{\alpha,\beta}\, \bs{z}, \\
    \Tr[G_\alpha\, Z\, H_\beta\, Z^\top] &= \bs{z}^\top (H_\beta \otimes G_\alpha)\, \bs{z} = \bs{z}^\top J_{\alpha,\beta}\, \bs{z},
\end{align}
which yields the vectorized form~\eqref{eq:support_function_of_Csep_as_a_quadratic_form}.
\end{proof}

\subsection{Proof of Theorem~\ref{thm:support_function_of_Cquantum}: Support function of $\Cquantum$}
\label{app:proof_support_Cquantum}
\begin{proof}
Since $\Cquantum$ is the image of the state space $\Squantum$ under the linear map $\Phi_{AB}:\hat{\rho} \mapsto (\Tr[\hat{\rho}\, (\hat{A}_i \otimes \hat{B}_j)])_{i,j}$, the support function equals the largest eigenvalue:
\begin{equation}
    \phiquantum(Z)
    = \sup_{\hat{\rho} \in \D(\CtwoCtwo)} \Tr[\hat{\rho}\, \hat{S}_{AB}(Z)]
    = \lambda_{\max}(\hat{S}_{AB}(Z)).
    \label{eq:app_phiqm_lambdamax}
\end{equation}
We now compute the spectrum of $\hat{S}_{AB}(Z)$.

\noindent\textbf{Step 1: Diagonalization via special SVD.}
Expanding each observable in the Pauli basis, we obtain
\begin{equation}
    \hat{S}_{AB}(Z) = \sum_{k,l} (A^\top Z B)_{kl}\, \hat{\sigma}_k \otimes \hat{\sigma}_l.
\end{equation}
Let $A^\top Z B = U \tilde{\Sigma} V^\top$ be the special singular value decomposition~\cite{sanyal2011orbitopes}, where $U, V \in \mrm{SO}(3)$ and $\tilde{\Sigma} = \mathrm{diag}(\tilde{s}_1, \tilde{s}_2, \tilde{s}_3)$ with
\begin{equation}
    \tilde{s}_i = s_i \quad (i=1,2), \qquad \tilde{s}_3 = s_3\, \op{sgn}(\det(A^\top Z B)).
    \label{eq:special_singular_values}
\end{equation}
Here $s_1 \geq s_2 \geq s_3 \geq 0$ are the ordinary singular values of $A^\top Z B$.

Since $U, V \in \mrm{SO}(3)$, there exist unitary operators $\hat{W}_A, \hat{W}_B$ on $\mbb{C}^2$ satisfying
\begin{equation}
    \hat{W}_A\, \hat{\sigma}_k\, \hat{W}_A^\dagger = \sum_{k'} U_{k'k}\, \hat{\sigma}_{k'}, \quad
    \hat{W}_B\, \hat{\sigma}_l\, \hat{W}_B^\dagger = \sum_{l'} V_{l'l}\, \hat{\sigma}_{l'}.
\end{equation}
The conjugation $(\hat{W}_A \otimes \hat{W}_B)^\dagger \hat{S}_{AB}(Z) (\hat{W}_A \otimes \hat{W}_B)$ preserves the spectrum, so we reduce to
\begin{equation}
    \hat{S}_{AB}'(Z) = \sum_{n=1}^{3} \tilde{s}_n\, \hat{\sigma}_n \otimes \hat{\sigma}_n.
    \label{eq:app_S_prime_diagonal}
\end{equation}

\noindent\textbf{Step 2: Eigenvalue calculation.}
The simultaneous eigenvalues of $(\hat{\sigma}_1 \otimes \hat{\sigma}_1,\, \hat{\sigma}_2 \otimes \hat{\sigma}_2,\, \hat{\sigma}_3 \otimes \hat{\sigma}_3)$ are $(\theta_1, \theta_2, \theta_3) \in \{\pm 1\}^3$ subject to the constraint $\theta_1 \theta_2 \theta_3 = -1$. The eigenvalues of $\hat{S}_{AB}'(Z)$ are therefore
\begin{equation}
    \tilde{s}_1 \theta_1 + \tilde{s}_2 \theta_2 + \tilde{s}_3 \theta_3, \qquad \theta_1 \theta_2 \theta_3 = -1.
\end{equation}
Since $\tilde{s}_1 \geq \tilde{s}_2 \geq |\tilde{s}_3|$, the largest eigenvalue is
\begin{equation}
    \lambda_{\max}(\hat{S}_{AB}(Z))
    = \tilde{s}_1 + \tilde{s}_2 - \tilde{s}_3
    = s_1 + s_2 - s_3\, \op{sgn}(\det(A^\top ZB)).
    \label{eq:app_general_formula_lambda_max}
\end{equation}

In the case $m=2$, the $3 \times 3$ matrix $A^\top ZB$ has rank at most~$2$, so $s_3 = 0$ and $\phiquantum(Z) = s_1 + s_2$.

\noindent\textbf{Step 3: Derivation of the explicit formula for $m=2$.}
By the same argument as in Appendix~\ref{app:proof_support_Csep}, the singular values $s_1, s_2$ satisfy
\begin{equation}
    s_1^2 + s_2^2 = \Tr[G_\alpha\, Z\, G_\beta\, Z^\top], \qquad
    s_1^2\, s_2^2 = \det(G_\alpha\, Z\, G_\beta\, Z^\top) = \sin^2\!\alpha\, \sin^2\!\beta\, (\det Z)^2.
\end{equation}
Therefore,
\begin{equation}
    \phiquantum(Z)
    = s_1 + s_2
    = \sqrt{s_1^2 + s_2^2 + 2 s_1 s_2}
    = \sqrt{\Tr[G_\alpha\, Z\, G_\beta\, Z^\top] + 2|\det Z|\sin\alpha\sin\beta},
\end{equation}
which is the matrix form~\eqref{eq:support_function_of_Cquantum_matrix_form}.
By Lemma~\ref{lem:vectorization}, $\Tr[G_\alpha\, Z\, G_\beta\, Z^\top] = \bs{z}^\top K_{\alpha,\beta}\, \bs{z}$, yielding the vectorized form~\eqref{eq:support_function_of_Cquantum_as_a_quadratic_form}.
\end{proof}

\subsection{Proof of Theorem~\ref{thm:support_function_of_Cmax}: Support function of $\Cmax$}
\label{app:proof_support_Cmax}
\begin{proof}
    Since $\Cmax$ is the image of $\Smax$ under the linear map $\hat{\rho} \mapsto (\Tr[\hat{\rho}\, \hat{A}_i \otimes \hat{B}_j])_{i,j}$, we compute the support function by maximizing over the extreme points of $\Smax$:
    \begin{equation}
        \phimax(Z) = \sup_{\hat{\rho} \in \mrm{ext}\left(\Smax\right)} \Tr[\hat{\rho}\, \hat{S}_{AB}(Z)].
    \end{equation}
    The extreme points of $\Smax$ consist of pure quantum states and partial transposes of pure entangled states \cite[Proposition~12]{arai2024detecting_supplemental}. Therefore,
    \begin{equation}
        \phimax(Z) = \max\left\{
            \sup_{\|\psi\|=1} \Tr\left[|\psi\ket \bra\psi|\, \hat{S}_{AB}(Z)\right],~
            \sup_{\|\psi\|=1} \Tr\left[\Gamma(|\psi\ket \bra\psi|)\, \hat{S}_{AB}(Z)\right]
        \right\},
    \end{equation}
    where $\Gamma := \mrm{id} \otimes \top$ denotes the partial transpose map.

    The first term equals the support function of $\Cquantum$:
    \begin{equation}
        \sup_{\|\psi\|=1} \Tr\left[|\psi\ket \bra\psi|\, \hat{S}_{AB}(Z)\right]
        = \phiquantum(Z) = s_1 + s_2 - s_3 \,\op{sgn} (\det(A^\top Z B)).
    \end{equation}

    For the second term, we use $\Tr\left[\hat{X}\Gamma(\hat{Y})\right] = \Tr\left[\Gamma(\hat{X}) \hat{Y}\right]$ for any operators $\hat{X}, \hat{Y}$ to obtain
    \begin{equation}
        \sup_{\|\psi\|=1} \Tr\left[\Gamma(|\psi\ket \bra\psi)\, \hat{S}_{AB}(Z)\right]
        = \sup_{\|\psi\|=1} \Tr\left[|\psi\ket \bra\psi|\, \Gamma(\hat{S}_{AB}(Z))\right]
        = \lambda_{\max}(\Gamma(\hat{S}_{AB}(Z))).
    \end{equation}
    Since $(\mrm{id} \otimes \top)(\hat{\sigma}_k \otimes \hat{\sigma}_l) = \hat{\sigma}_k \otimes \hat{\sigma}_l^{*}$ and $\hat{\sigma}_l^{*} = R_{ll}\, \hat{\sigma}_l$ with $R := \mathrm{diag}(1,-1,1)$, we have
    \begin{equation}
        \Gamma(\hat{S}_{AB}(Z)) = \sum_{k,l} (A^\top Z B\, R)_{kl}\, \hat{\sigma}_k \otimes \hat{\sigma}_l.
    \end{equation}
    Since $\det R = -1$, the special singular values~\eqref{eq:special_singular_values} of $A^\top Z B\, R$ are $s_1, s_2, -s_3\, \op{sgn}(\det(A^\top ZB))$. By the same argument as in Appendix~\ref{app:proof_support_Cquantum}, the largest eigenvalue of $\Gamma(\hat{S}_{AB}(Z))$ is
    \begin{equation}
        \lambda_{\max}(\Gamma(\hat{S}_{AB}(Z)))
        = s_1 + s_2 + s_3\, \op{sgn}(\det(A^\top Z B)).
    \end{equation}

    Combining the two terms, we obtain
    \begin{align*}
        \phimax(Z) &= \max\{s_1 + s_2 - s_3\, \op{sgn}(\det(A^\top Z B)),~ s_1 + s_2 + s_3\, \op{sgn}(\det(A^\top Z B))\} \\
        &= s_1 + s_2 + s_3 \\
        &= \|A^\top Z B\|_{1}. \numberthis
    \end{align*}
\end{proof}

\section{Proofs of the Convex Hull Characterizations}\label{appsec:proof_convex_hull}

All three proofs follow the same strategy: identify the support function of the candidate convex hull with the known support function of the correlation set, and conclude equality of the two convex sets.
A closed convex compact set is uniquely determined by its support function, so it suffices to show that the support functions agree for all $Z$.

\subsection{Proof of Theorem~\ref{thm:convex_hull_characterization_of_Csep}: Convex hull characterization of $\Csep$}
\label{app:proof_convex_hull_Csep}
\begin{proof}
    Let $K := \op{conv}\{A\, \bs{r}_A\bs{r}_B^\top\, B^\top : \bs{r}_A, \bs{r}_B \in \mbb{R}^3,\, \|\bs{r}_A\| = \|\bs{r}_B\| = 1\}$.
    The support function of $K$ is
    \begin{equation}
        \phi_K(Z) = \sup_{C \in K} \Tr[Z^\top C]
        = \max_{\|\bs{r}_A\| = \|\bs{r}_B\| = 1} \Tr[Z^\top A\, \bs{r}_A\bs{r}_B^\top\, B^\top]
        = \max_{\|\bs{r}_A\| = \|\bs{r}_B\| = 1} \bs{r}_A^\top (A^\top Z B)\, \bs{r}_B
        = \|A^\top Z B\|_{\infty},
    \end{equation}
    where the second equality uses the fact that the support function of a convex hull equals the supremum over the generating set, and the last equality is the variational characterization of the operator norm.
    By Theorem~\ref{thm:support_function_of_Csep}, $\phisep(Z) = \|A^\top Z B\|_{\infty} = \phi_K(Z)$ for all $Z$, so $\Csep = K$.
\end{proof}

\subsection{Proof of Theorem~\ref{thm:convex_hull_characterization_of_Cquantum}: Convex hull characterization of $\Cquantum$}
\label{app:proof_convex_hull_Cquantum}
\begin{proof}
    Let $K := \op{conv}\{A Q B^\top : Q \in \mrm{SO}^{-}(3)\}$.
    The support function of $K$ is
    \begin{equation}
        \phi_K(Z) = \max_{Q \in \mrm{SO}^{-}(3)} \Tr[Z^\top (A Q B^\top)]
        = \max_{Q \in \mrm{SO}^{-}(3)} \Tr[(A^\top Z B)^\top Q].
    \end{equation}
    Setting $X := A^\top Z B$ with singular values $s_1 \geq s_2 \geq s_3 \geq 0$, we evaluate the right-hand side.
    The optimal value of the Procrustes problem $\max_{R \in \mrm{SO}(3)} \Tr[X^\top R]$ equals $s_1 + s_2 + s_3\, \op{sgn}(\det X)$~\cite{umeyama1991leastsquares}. Since any $Q \in \mrm{SO}^{-}(3)$ can be written as $Q = R J$ with $R \in \mrm{SO}(3)$ and $J = \op{diag}(1,1,-1)$, we have
    \begin{align}
        \max_{Q \in \mrm{SO}^{-}(3)} \Tr[X^\top Q]
        &= \max_{R \in \mrm{SO}(3)} \Tr[(XJ)^\top R] \notag \\
        &= s_1 + s_2 + s_3\, \op{sgn}(\det(XJ)) \notag \\
        &= s_1 + s_2 - s_3\, \op{sgn}(\det X),
        \label{eq:procrustes_SO_minus}
    \end{align}
    where we used the fact that the singular values of $XJ$ coincide with those of $X$, and $\det(XJ) = -\det X$.
    By Theorem~\ref{thm:support_function_of_Cquantum}, $\phiquantum(Z) = s_1 + s_2 - s_3\, \op{sgn}(\det(A^\top Z B)) = \phi_K(Z)$ for all $Z$, so $\Cquantum = K$.
\end{proof}

\subsection{Proof of Theorem~\ref{thm:convex_hull_characterzation_of_Cmax}: Convex hull characterization of $\Cmax$}
\label{app:proof_convex_hull_Cmax}
\begin{proof}
    Let $K := \op{conv}\{A Q B^\top : Q \in \mrm{O}(3)\}$.
    The support function of $K$ is
    \begin{equation}
        \phi_K(Z) = \max_{Q \in \mrm{O}(3)} \Tr[Z^\top (A Q B^\top)]
        = \max_{Q \in \mrm{O}(3)} \Tr[(A^\top Z B)^\top Q]
        = \|A^\top Z B\|_{1},
    \end{equation}
    where the last equality is the variational characterization of the trace norm: $\|X\|_1 = \max_{Q \in \mrm{O}(n)} \Tr[X^\top Q]$ (e.g., Theorem 3.4.1 in \cite{Horn1994topics}).
    By Theorem~\ref{thm:support_function_of_Cmax}, $\phimax(Z) = \|A^\top Z B\|_1 = \phi_K(Z)$ for all $Z$, so $\Cmax = K$.
\end{proof}

\section{Proofs of the Gauge Function Theorems}\label{appsec:proofs_gauge_functions}

Throughout this section, $P_M$ denotes the orthogonal projection onto a subspace $M$.

\begin{lemma}\label{lem:range_condition}
    For any $C \in \MmR$, the following conditions are equivalent:
    \begin{enumerate}
        \item $C \in \ran(Y \mapsto A Y B^\top)$, i.e., there exists $Y \in \mcal{M}_3(\mbb{R})$ such that $C = A Y B^\top$.
        \item $\ran C \subset \ran A$ and $\ran C^\top \subset \ran B$.
    \end{enumerate}
    When these conditions hold, $Y := A^{+} C\, (B^\top)^{+}$ satisfies $A Y B^\top = C$.
\end{lemma}
\begin{proof}
    (i)~$\Rightarrow$~(ii): If $C = AYB^\top$, then $\ran C \subset \ran A$ and $\ran C^\top \subset \ran B$.

    (ii)~$\Rightarrow$~(i): Suppose $\ran C \subset \ran A$ and $\ran C^\top \subset \ran B$. Since $AA^{+}=P_{\ran A}$ is the orthogonal projection onto $\ran A$, the condition $\ran C \subset \ran A$ implies $AA^{+} C = C$. Similarly, since $(B^\top)^{+} B^\top = P_{\ran B}$, $\ran C^\top \subset \ran B$ implies $C (B^\top)^{+} B^\top = C$. Therefore, $Y := A^{+} C\, (B^\top)^{+}$ satisfies
    \begin{equation}
        AYB^\top = A A^{+} C\, (B^\top)^{+} B^\top = C. \qedhere
    \end{equation}
\end{proof}

\begin{lemma}[{Duality of the asymmetric norms $\|\cdot\|_{\pm}$}]\label{lem:asymmetric_norm_duality}
    For any $X \in \mcal{M}_3(\mbb{R})$,
    \begin{equation}
        \|X\|_{+} = \sup_{Y \neq 0} \frac{\Tr[X^\top Y]}{\|Y\|_{-}}, \qquad
        \|X\|_{-} = \sup_{Y \neq 0} \frac{\Tr[X^\top Y]}{\|Y\|_{+}}.
    \end{equation}
\end{lemma}
\begin{proof}
\noindent\textbf{Proof of the first identity.}

\noindent\emph{Step 1: Identification of the unit ball.}
By~\cite{saunderson2015semidefinite}, a matrix $Y \in \mcal{M}_3(\mbb{R})$ with singular values $s_1 \geq s_2 \geq s_3 \geq 0$ belongs to $\op{conv}(\mrm{SO}(3))$ if and only if
\begin{equation}\label{eq:saunderson_SO3_full}
    \max\{\tau_1 + \tau_2 - \tau_3,\; -\tau_1 + \tau_2 + \tau_3,\; \tau_1 - \tau_2 + \tau_3,\; -\tau_1 - \tau_2 - \tau_3\} \leq 1,
\end{equation}
where $\tau_1 := s_1,\, \tau_2 := s_2,\, \tau_3 := s_3\,\op{sgn}(\det Y)$ are the signed singular values of $Y$.
The left-hand side of~\eqref{eq:saunderson_SO3_full} simplifies by cases:
if $\det Y \geq 0$, then $\tau_3 \geq 0$, and the maximum equals $\tau_1 + \tau_2 - \tau_3 = s_1 + s_2 - s_3$;
if $\det Y < 0$, then $\tau_3 < 0$, and the maximum equals $\tau_1 + \tau_2 - \tau_3 = s_1 + s_2 + s_3$.
In both cases, the condition~\eqref{eq:saunderson_SO3_full} reduces to
\begin{equation}\label{eq:saunderson_SO3}
    s_1 + s_2 - s_3\,\op{sgn}(\det Y) \leq 1,
\end{equation}
i.e., $\|Y\|_{-} \leq 1$.
Therefore, $\op{conv}(\mrm{SO}(3)) = \{Y \in \mcal{M}_3(\mbb{R}) : \|Y\|_{-} \leq 1\}$, and $\|\cdot\|_{-}$ is the gauge function of $\op{conv}(\mrm{SO}(3))$.

\noindent\emph{Step 2: Support function of $\op{conv}(\mrm{SO}(3))$.}
We show that the support function of $\op{conv}(\mrm{SO}(3))$ is $\|\cdot\|_{+}$.
Let $X = U \Sigma V^\top$ be a singular value decomposition with $\Sigma = \op{diag}(s_1, s_2, s_3)$.
For any $Q \in \mrm{SO}(3)$, setting $R := V^\top Q U$, we have $\det R = \det(V^\top U)$ and
\begin{equation}
    \Tr[X^\top Q] = \Tr[\Sigma\, R].
\end{equation}
If $\det(V^\top U) = 1$, then $R$ ranges over $\mrm{SO}(3)$ and $\max \Tr[\Sigma\, R] = s_1 + s_2 + s_3$, attained at $R = I$.
If $\det(V^\top U) = -1$, then $R$ ranges over $\mrm{SO}^{-}(3)$ and $\max \Tr[\Sigma\, R] = s_1 + s_2 - s_3$, attained at $R = \op{diag}(1,1,-1)$.
Since $\op{sgn}(\det X) = \op{sgn}(\det(V^\top U))$ when $s_3 > 0$, and $s_1 + s_2 - s_3 = s_1 + s_2 + s_3$ when $s_3 = 0$, we obtain in all cases
\begin{equation}
    \max_{Q \in \mrm{SO}(3)} \Tr[X^\top Q] = s_1 + s_2 + s_3\,\op{sgn}(\det X) = \|X\|_{+}.
\end{equation}

\noindent\emph{Step 3: Duality.}
By Steps~1 and~2, $\|\cdot\|_{-}$ is the gauge function and $\|\cdot\|_{+}$ is the support function of $\op{conv}(\mrm{SO}(3))$.
Since $\op{conv}(\mrm{SO}(3))$ is a convex compact set containing the origin in its interior,
the general duality relation~\eqref{eq:gauge_function_as_a_dual_of_support_function} gives
\begin{equation}
    \|X\|_{+} = \sup_{Y \neq 0} \frac{\Tr[X^\top Y]}{\|Y\|_{-}}.
\end{equation}

\noindent\textbf{Proof of the second identity.}
Since $\mrm{SO}^{-}(3) = J\, \mrm{SO}(3)$ with $J := \op{diag}(1,1,-1)$, we have $\op{conv}(\mrm{SO}^{-}(3)) = J\, \op{conv}(\mrm{SO}(3))$.
Therefore, $Y \in \op{conv}(\mrm{SO}^{-}(3))$ if and only if $JY \in \op{conv}(\mrm{SO}(3))$, i.e., $\|JY\|_{-} \leq 1$.
Since $JY$ has the same singular values as $Y$ and $\det(JY) = -\det Y$, we have
\begin{equation}
    \|JY\|_{-} = s_1 + s_2 - s_3\,\op{sgn}(\det(JY)) = s_1 + s_2 + s_3\,\op{sgn}(\det Y) = \|Y\|_{+}.
\end{equation}
Therefore, $\op{conv}(\mrm{SO}^{-}(3)) = \{Y \in \mcal{M}_3(\mbb{R}) : \|Y\|_{+} \leq 1\}$, and $\|\cdot\|_{+}$ is the gauge function of $\op{conv}(\mrm{SO}^{-}(3))$.

For the support function, the same SVD argument as in Step~2 gives, for $X = U\Sigma V^\top$,
\begin{equation}
    \Tr[X^\top Q] = \Tr[\Sigma\, R], \qquad R := V^\top Q U.
\end{equation}
When $Q$ ranges over $\mrm{SO}^{-}(3)$, the matrix $R$ ranges over orthogonal matrices with $\det R = -\det(V^\top U)$.
If $\det(V^\top U) = 1$, then $R$ ranges over $\mrm{SO}^{-}(3)$ and $\max \Tr[\Sigma\, R] = s_1 + s_2 - s_3$.
If $\det(V^\top U) = -1$, then $R$ ranges over $\mrm{SO}(3)$ and $\max \Tr[\Sigma\, R] = s_1 + s_2 + s_3$.
In all cases,
\begin{equation}
    \max_{Q \in \mrm{SO}^{-}(3)} \Tr[X^\top Q] = s_1 + s_2 - s_3\,\op{sgn}(\det X) = \|X\|_{-}.
\end{equation}
The duality relation~\eqref{eq:gauge_function_as_a_dual_of_support_function} applied to $\op{conv}(\mrm{SO}^{-}(3))$ then gives
\begin{equation}
    \|X\|_{-} = \sup_{Y \neq 0} \frac{\Tr[X^\top Y]}{\|Y\|_{+}}. \qedhere
\end{equation}
\end{proof}

\subsection{Proof of Theorem~\ref{thm:gauge_function_of_Csep}: Gauge function of $\Csep$}
\label{app:proof_gauge_Csep}
\begin{proof}
\noindent\textbf{Step 1: Derivation via duality.}
By the duality relation~\eqref{eq:gauge_function_as_a_dual_of_support_function} and Theorem~\ref{thm:support_function_of_Csep}, the gauge function $\gammasep$ satisfies
\begin{equation}
    \gammasep(C) = \sup_{Z \neq 0} \frac{\Tr[Z^\top C]}{\phisep(Z)}
    = \sup_{Z \neq 0} \frac{\Tr[Z^\top C]}{\|A^\top Z\, B\|_{\infty}}.
\end{equation}
We denote $W := A^{+} C\, (B^\top)^{+}$ throughout this proof.

When $\ran C \subset \ran A$ and $\ran C^\top \subset \ran B$, the substitution $\tilde{Z} := A^\top Z\, B$ gives
\begin{equation}
    \Tr[Z^\top C] = \Tr[\tilde{Z}^\top\, W],
    \label{eq:trace_substitution_sep}
\end{equation}
since $(A^\top)^{+} \tilde{Z}\, B^{+} = P_{\ran A}\, Z\, P_{\ran B}$ and the range conditions imply $P_{\ran A}\, C\, P_{\ran B} = C$.
The H\"older inequality for Schatten norms, $|\Tr[\tilde{Z}^\top W]| \leq \|\tilde{Z}\|_{\infty}\, \|W\|_1$, then yields the upper bound
\begin{equation}
    \frac{\Tr[Z^\top C]}{\|A^\top Z\, B\|_{\infty}} = \frac{\Tr[\tilde{Z}^\top\, W]}{\|\tilde{Z}\|_{\infty}} \leq \|W\|_1.
    \label{eq:gauge_sep_upper_bound}
\end{equation}

We show that this bound is attained.
Let $W = U \Sigma V^\top$ be a singular value decomposition, and define $Z_{\star} := (A^\top)^{+}\, U V^\top\, B^{+}$.
Setting $\tilde{Z}_{\star} := A^\top Z_{\star}\, B$, we obtain
\begin{equation}
    \tilde{Z}_{\star} = P_{\ran A^\top}\, U V^\top\, P_{\ran B^\top},
\end{equation}
where we used $A^\top (A^\top)^{+} = P_{\ran A^\top}$ and $B^{+} B = P_{\ran B^\top}$.
Since orthogonal projections are contractive in the operator norm,
\begin{equation}
    \|\tilde{Z}_{\star}\|_{\infty} \leq \|U V^\top\|_{\infty} = 1.
    \label{eq:Ztilde_norm_bound}
\end{equation}
For the numerator, the range inclusions $\ran W \subset \ran A^\top$ and $\ran W^\top \subset \ran B^\top$ imply $P_{\ran A^\top}\, W = W$ and $W\, P_{\ran B^\top} = W$, so
\begin{equation}
    \Tr[\tilde{Z}_{\star}^\top\, W]
    = \Tr[V U^\top\, P_{\ran A^\top}\, W\, P_{\ran B^\top}]
    = \Tr[V U^\top\, U \Sigma V^\top]
    = \Tr[\Sigma]
    = \|W\|_1.
    \label{eq:Ztilde_trace_value}
\end{equation}
Combining~\eqref{eq:gauge_sep_upper_bound}--\eqref{eq:Ztilde_trace_value}, we conclude
\begin{equation}
    \gammasep(C) \geq \frac{\Tr[\tilde{Z}_{\star}^\top\, W]}{\|\tilde{Z}_{\star}\|_{\infty}} \geq \|W\|_1,
\end{equation}
and therefore $\gammasep(C) = \|A^{+} C\, (B^\top)^{+}\|_1$, with the supremum attained at $Z = Z_{\star}$.

If $\ran C \not\subset \ran A$ or $\ran C^\top \not\subset \ran B$, there exists $Z$ in the kernel of the map $Z \mapsto A^\top Z B$ with $\Tr[Z^\top C] \neq 0$. Then $\phisep(Z) = 0$ while $\Tr[Z^\top C] > 0$ for an appropriate sign choice, giving $\gammasep(C) = +\infty$.

\noindent\textbf{Step 2: Explicit formula for $m=2$.}
When $\sin\alpha, \sin\beta \neq 0$, the Gram matrices $G_\alpha = AA^\top$ and $G_\beta = BB^\top$ are invertible, and the pseudoinverses reduce to
\begin{equation}
    A^{+} = A^\top G_\alpha^{-1}, \qquad (B^\top)^{+} = G_\beta^{-1} B.
\end{equation}
The $3 \times 3$ matrix $W := A^{+} C\, (B^\top)^{+} = A^\top G_\alpha^{-1} C\, G_\beta^{-1} B$ has rank at most~$2$, so $s'_3 = 0$ and $\gammasep(C) = s'_1 + s'_2$.
The nonzero squared singular values $s_1'^2, s_2'^2$ are the nonzero eigenvalues of $W^\top W = B^\top G_\beta^{-1} C^\top G_\alpha^{-1} C\, G_\beta^{-1} B$.
Since $XY$ and $YX$ share the same nonzero eigenvalues, these coincide with the eigenvalues of the $2 \times 2$ matrix
\begin{equation}
    Q := G_\alpha^{-1} C\, G_\beta^{-1} C^\top.
\end{equation}
Therefore,
\begin{equation}
    s_1'^2 + s_2'^2 = \Tr Q = \Tr[G_\alpha^{-1} C\, G_\beta^{-1} C^\top],
    \qquad
    s_1'^2\, s_2'^2 = \det Q = \frac{(\det C)^2}{\sin^2\!\alpha\, \sin^2\!\beta}.
\end{equation}
By Lemma~\ref{lem:vectorization},
\begin{equation}
    \Tr[G_\alpha^{-1} C\, G_\beta^{-1} C^\top]
    = \bs{c}^\top (G_\beta^{-1} \otimes G_\alpha^{-1})\, \bs{c}
    = \bs{c}^\top K_{\alpha,\beta}^{-1}\, \bs{c}.
\end{equation}
Since $\gammasep(C) = s'_1 + s'_2 = \sqrt{s_1'^2 + s_2'^2 + 2 s'_1 s'_2}$, we obtain
\begin{equation}
    \gammasep(C) = \sqrt{\bs{c}^\top K_{\alpha,\beta}^{-1}\, \bs{c} + \frac{2|\det C|}{\sin\alpha\, \sin\beta}}
    = \sqrt{\Tr[G_\alpha^{-1}\, C\, G_\beta^{-1}\, C^\top] + \frac{2|\det C|}{\sin\alpha\, \sin\beta}}.
\end{equation}
\end{proof}

\subsection{Proof of Theorem~\ref{thm:gauge_function_of_Cquantum}: Gauge function of $\Cquantum$}
\label{app:proof_gauge_Cquantum}
\begin{proof}
\noindent\textbf{Case $m\geq 3$.}
We follow the same strategy as the proofs of Theorems~\ref{thm:gauge_function_of_Csep} and~\ref{thm:gauge_function_of_Cmax}.
By the duality relation~\eqref{eq:gauge_function_as_a_dual_of_support_function} and Theorem~\ref{thm:support_function_of_Cquantum}, the gauge function $\gammaquantum$ satisfies
\begin{equation}
    \gammaquantum(C) = \sup_{Z \neq 0} \frac{\Tr[Z^\top C]}{\phiquantum(Z)}
    = \sup_{Z \neq 0} \frac{\Tr[Z^\top C]}{\|A^\top Z\, B\|_{-}}.
\end{equation}
We denote $W := A^{+} C\, (B^\top)^{+}$.

If $\ran C \not\subset \ran A$ or $\ran C^\top \not\subset \ran B$, there exists $Z$ in the kernel of the map $Z \mapsto A^\top Z B$ with $\Tr[Z^\top C] \neq 0$. Then $\phiquantum(Z) = 0$ while $\Tr[Z^\top C] > 0$ for an appropriate sign choice, giving $\gammaquantum(C) = +\infty$.

For the remainder of the proof, we assume $\ran C \subset \ran A$ and $\ran C^\top \subset \ran B$.
The substitution $\tilde{Z} := A^\top Z\, B$ and the identity~\eqref{eq:trace_substitution_sep} give
\begin{equation}
    \frac{\Tr[Z^\top C]}{\|A^\top Z\, B\|_{-}} = \frac{\Tr[\tilde{Z}^\top\, W]}{\|\tilde{Z}\|_{-}} \leq \|W\|_{+},
    \label{eq:gauge_qm_upper_bound}
\end{equation}
where the inequality follows from Lemma~\ref{lem:asymmetric_norm_duality}.

\noindent\textbf{Subcase $r = 3$.}
We show that the upper bound~\eqref{eq:gauge_qm_upper_bound} is attained.
Let $W = U \Sigma V^\top$ be a singular value decomposition with $\Sigma = \op{diag}(s'_1, s'_2, s'_3)$, and set $\epsilon' := \op{sgn}(\det W)$.
Define $Z_{\star} := (A^\top)^{+}\, U\, \op{diag}(1,1,\epsilon')\, V^\top\, B^{+}$.
Since $\rank A = \rank B = 3$, the projections $P_{\ran A^\top}$ and $P_{\ran B^\top}$ are the identity on $\mbb{R}^3$, so
\begin{equation}
    \tilde{Z}_{\star} := A^\top Z_{\star}\, B = U\, \op{diag}(1,1,\epsilon')\, V^\top.
\end{equation}
When $\epsilon' = \pm 1$, the singular values of $\tilde{Z}_{\star}$ are $1,1,1$, and $\det(\tilde{Z}_{\star}) = \epsilon'\,\det(U V^\top) = \epsilon'^2 = 1$ (using $\epsilon' = \op{sgn}(\det(U V^\top))$ when $s'_3 > 0$). Therefore,
\begin{equation}
    \|\tilde{Z}_{\star}\|_{-} = 1 + 1 - 1 = 1.
\end{equation}
When $\epsilon' = 0$ (i.e., $s'_3 = 0$), the singular values of $\tilde{Z}_{\star}$ are $1,1,0$, and $\|\tilde{Z}_{\star}\|_{-} = 1 + 1 = 2$.
In this case, we instead set $Z_{\star} := (A^\top)^{+}\, U\, \op{diag}(1,1,\eta)\, V^\top\, B^{+}$ with $\eta := \op{sgn}(\det(UV^\top)) = \pm 1$, giving $\tilde{Z}_{\star} = U\, \op{diag}(1,1,\eta)\, V^\top$ with singular values $1,1,1$ and $\det(\tilde{Z}_{\star}) = 1$, so $\|\tilde{Z}_{\star}\|_{-} = 1$.

For the numerator, we compute
\begin{equation}
    \Tr[\tilde{Z}_{\star}^\top\, W]
    = \Tr[V\, \op{diag}(1,1,\epsilon')\, U^\top\, U \Sigma V^\top]
    = \Tr[\op{diag}(1,1,\epsilon')\, \Sigma]
    = s'_1 + s'_2 + \epsilon' s'_3
    = \|W\|_{+},
\end{equation}
where in the case $\epsilon' = 0$ we used $\eta\, s'_3 = 0 = \epsilon'\, s'_3$.
Therefore, $\gammaquantum(C) \geq \|W\|_{+}/1 = \|W\|_{+}$, and combined with~\eqref{eq:gauge_qm_upper_bound}, we conclude $\gammaquantum(C) = \|A^{+} C\, (B^\top)^{+}\|_{+}$.

\noindent\textbf{Subcase $r \leq 2$.}
Since $\Cquantum = \Cmax$, the gauge function of $\Cquantum$ coincides with that of $\Cmax$ (Theorem~\ref{thm:gauge_function_of_Cmax}).

\noindent\textbf{Case $m=2$.}
Since $\Cquantum$ is a convex compact set containing the origin in its interior, the gauge function uniquely determines the set:
\begin{equation}
    \Cquantum = \{C \in \mcal{M}_2(\mathbb{R}) \,:\, \gammaquantum(C) \leq 1\}.
\end{equation}
Define
\begin{equation}
    p_{A,B}(C) := \frac{1}{2}\left(\sqrt{\bs{c}^\top F_{\alpha,\beta}\, \bs{c}} + \sqrt{\bs{c}^\top F_{\alpha,-\beta}\, \bs{c}}\right).
\end{equation}
By the necessary and sufficient condition for quantum realizability~\cite{nogami2025necessary},
\begin{equation}
    \Cquantum = \{C \in \mcal{M}_2(\mathbb{R}) \,:\, p_{A,B}(C) \leq 1\}.
\end{equation}
Since both $\gammaquantum$ and $p_{A,B}$ are seminorms on $\mcal{M}_2(\mathbb{R})$ and the unit ball of a seminorm uniquely determines the seminorm itself, we conclude $\gammaquantum = p_{A,B}$.

We now show that $p_{A,B}$ admits the matrix form~\eqref{eq:gauge_function_of_Cquantum_matrix_form}.
By Lemma~\ref{lem:vectorization},
\begin{equation}
    \Tr[L_\alpha\, C\, L_\beta^\top\, C^\top]
    = \bs{c}^\top (L_\beta \otimes L_\alpha)\, \bs{c},
\end{equation}
where $L_\beta \otimes L_\alpha = F_{\alpha,\beta} + \im\, \tilde{F}_{\alpha,\beta}$ with $F_{\alpha,\beta} := \op{Re}(L_\beta \otimes L_\alpha)$ and $\tilde{F}_{\alpha,\beta} := \op{Im}(L_\beta \otimes L_\alpha)$.
A direct computation shows that $\tilde{F}_{\alpha,\beta}$ is antisymmetric, so $\bs{c}^\top \tilde{F}_{\alpha,\beta}\, \bs{c} = 0$ for any real vector $\bs{c}$.
Therefore,
\begin{equation}
    \Tr[L_\alpha\, C\, L_\beta^\top\, C^\top]
    = \bs{c}^\top F_{\alpha,\beta}\, \bs{c},
\end{equation}
and similarly for $\beta \to -\beta$, which gives~\eqref{eq:gauge_function_of_Cquantum_matrix_form}.
\end{proof}

\subsection{Proof of Theorem~\ref{thm:gauge_function_of_Cmax}: Gauge function of $\Cmax$}
\label{app:proof_gauge_Cmax}
\begin{proof}
We follow the same strategy as Step~1 of the proof of Theorem~\ref{thm:gauge_function_of_Csep}.
By the duality relation~\eqref{eq:gauge_function_as_a_dual_of_support_function} and Theorem~\ref{thm:support_function_of_Cmax}, the gauge function $\gammamax$ satisfies
\begin{equation}
    \gammamax(C) = \sup_{Z \neq 0} \frac{\Tr[Z^\top C]}{\phimax(Z)}
    = \sup_{Z \neq 0} \frac{\Tr[Z^\top C]}{\|A^\top Z\, B\|_{1}}.
\end{equation}
When $\ran C \subset \ran A$ and $\ran C^\top \subset \ran B$, the substitution $\tilde{Z} := A^\top Z\, B$ and the identity~\eqref{eq:trace_substitution_sep} give
\begin{equation}
    \frac{\Tr[Z^\top C]}{\|A^\top Z\, B\|_{1}} = \frac{\Tr[\tilde{Z}^\top\, W]}{\|\tilde{Z}\|_{1}} \leq \|W\|_{\infty},
\end{equation}
where $W := A^{+} C\, (B^\top)^{+}$ and the inequality is the H\"older inequality for Schatten norms.
Let $W = U \Sigma V^\top$ be a singular value decomposition with $s'_1 = \|W\|_\infty$, and let $\bs{u}_1, \bs{v}_1$ denote the first columns of $U, V$.
Setting $Z_{\star} := (A^\top)^{+}\, \bs{u}_1 \bs{v}_1^\top\, B^{+}$, we have $\tilde{Z}_{\star} = P_{\ran A^\top}\, \bs{u}_1 \bs{v}_1^\top\, P_{\ran B^\top}$.
Since $\bs{u}_1 \in \ran W \subset \ran A^\top$ and $\bs{v}_1 \in \ran W^\top \subset \ran B^\top$, we obtain $\tilde{Z}_{\star} = \bs{u}_1 \bs{v}_1^\top$, so $\|\tilde{Z}_{\star}\|_1 = 1$ and $\Tr[\tilde{Z}_{\star}^\top W] = s'_1 = \|W\|_\infty$.
Therefore $\gammamax(C) = \|A^{+} C\, (B^\top)^{+}\|_{\infty}$.

If $\ran C \not\subset \ran A$ or $\ran C^\top \not\subset \ran B$, there exists $Z$ in the kernel of the map $Z \mapsto A^\top Z B$ with $\Tr[Z^\top C] \neq 0$. Then $\phimax(Z) = 0$ while $\Tr[Z^\top C] > 0$ for an appropriate sign choice, giving $\gammamax(C) = +\infty$.
\end{proof}

\section{Proofs of Maximum Ratios of Support Functions}\label{appsec:proofs_maximum_ratios}

\subsection{Achievable singular values of $A^\top Z B$}
\begin{lemma}{\label{lem:achivable_singular_values}}
    Let $r := \min\{ \rank A, \rank B \}$. For any $s_1 \geq \cdots \geq s_r \geq 0$, there exists $Z \in \MmR$ such that the singular values of $A^\top Z B$ are $s_1, \ldots, s_r$ with the remaining singular values equal to zero. Furthermore, the sign of $\det(A^\top Z B)$ can be chosen arbitrarily.
\end{lemma}
\begin{proof}
    Since $\rank A \geq r$, we choose orthonormal vectors $\bs{u}_1,\ldots,\bs{u}_m$ with $\bs{u}_1,\ldots, \bs{u}_r \in \ran A^\top$ and $\bs{u}_{r+1},\ldots,\bs{u}_m \in (\ran A^\top)^\perp$. Similarly, we choose orthonormal vectors $\bs{v}_1,\ldots,\bs{v}_m$ with $\bs{v}_1,\ldots, \bs{v}_r \in \ran B^\top$.

    Since $XX^{+}$ is the orthogonal projection onto $\ran X$, we have
    \begin{align}
        A^\top (A^\top)^{+} \bs{u}_i &= \begin{cases}
            \bs{u}_i & \text{if~} i \leq r, \\
            0 & \text{if~} i > r,
        \end{cases}
        \\
        B^\top (B^\top)^{+} \bs{v}_i &= \begin{cases}
            \bs{v}_i & \text{if~} i \leq r, \\
            0 & \text{if~} i > r.
        \end{cases}
    \end{align}
    Define $U := (\bs{u}_1, \ldots, \bs{u}_m)$, $V := (\bs{v}_1, \ldots, \bs{v}_m)$, and $\Sigma := \mathrm{diag}(s_1, \ldots, s_r, 0, \ldots, 0)$. Then
    \begin{equation}
        A^\top (A^\top)^{+}\, U\Sigma V^\top\, B^{+} B = U\Sigma V^\top.
    \end{equation}
    Setting $Z := (A^\top)^{+} U\Sigma V^\top B^{+}$ gives $A^\top Z B = U\Sigma V^\top$, which has the desired singular values.

    Finally, the sign of $\det(A^\top Z B) = \det(U) \det(\Sigma) \det(V)$ can be flipped by replacing $U$ with $U' := U \mathrm{diag}(1,\ldots,1,-1)$ if necessary.
\end{proof}

\subsection{Proof of Theorem~\ref{thm:maximum_ratio_of_support_functions}}
\label{app:proof_maximum_ratio_sep}
\begin{proof}
By Theorems~\ref{thm:support_function_of_Cquantum} and~\ref{thm:support_function_of_Csep}, the ratio of support functions equals
\begin{equation}
    \frac{\phiquantum(Z)}{\phisep(Z)}
    = \frac{s_1 + s_2 - \epsilon\,s_3 }{s_1},
\end{equation}
where $s_1 \geq s_2 \geq s_3 \geq 0$ are the singular values of $A^\top Z B$, and $\epsilon := \op{sgn}(\det(A^\top Z B))$. 

\noindent\textbf{When $r=3$:} 
Since $s_1 \geq s_2 \geq s_3 \geq 0$, we have
\begin{equation}
     \frac{\phiquantum(Z)}{\phisep(Z)}
    \leq \frac{s_1 + s_1 + s_1}{s_1} = 3.
\end{equation}
Equality is attained by choosing $s_1 = s_2 = s_3 > 0$ and $\det(A^\top Z B) < 0$, which is possible by Lemma~\ref{lem:achivable_singular_values}.

\noindent\textbf{When $r=2$:}
Since $s_1 \geq s_2 \geq 0$ and $s_3 = 0$, we have
\begin{equation}
     \frac{\phiquantum(Z)}{\phisep(Z)}
    = \frac{s_1 + s_2}{s_1} \leq \frac{s_1 + s_1}{s_1} = 2.
\end{equation}
Equality is attained by choosing $s_1 = s_2 > 0$, which is possible by Lemma~\ref{lem:achivable_singular_values}.

\noindent\textbf{When $r=1$:}
Since $s_1 > 0$ and $s_2 = s_3 = 0$, the ratio
\begin{equation}
     \frac{\phiquantum(Z)}{\phisep(Z)}
    = \frac{s_1}{s_1} = 1
\end{equation}
holds for all $Z \in \MmR$ with $\phisep(Z) > 0$.
\end{proof}

\subsection{Proof of Theorem~\ref{thm:maximum_ratio_of_support_functions_beyond_quantum}}
\label{app:proof_maximum_ratio_bqs}
\begin{proof}
By Theorems~\ref{thm:support_function_of_Cquantum} and~\ref{thm:support_function_of_Cmax}, the ratio of support functions equals
\begin{equation}
    \frac{\phimax (Z)}{\phiquantum(Z)}
    = \frac{s_1 + s_2 + s_3}{s_1 + s_2 - \epsilon\,s_3 },
\end{equation}
where $s_1 \geq s_2 \geq s_3 \geq 0$ are the singular values of $A^\top Z B$, and $\epsilon := \op{sgn}(\det(A^\top Z B))$. 

\noindent\textbf{When $r=3$:}
Since $s_1 \geq s_2 \geq s_3 \geq 0$, we have
\begin{equation}
    \frac{\phimax (Z)}{\phiquantum(Z)}
    \leq \frac{s_1 + s_2 + s_3}{s_1 + s_2 - s_3} \leq \frac{s_1 + s_1 + s_1}{s_1 + s_3 - s_3} = 3.
\end{equation}
Equality is attained by choosing $s_1 = s_2 = s_3 > 0$ and $\det(A^\top Z B) > 0$, which is possible by Lemma~\ref{lem:achivable_singular_values}.

\noindent\textbf{When $r \leq 2$:}
Since $s_3 = 0$, the ratio
\begin{equation}
    \frac{\phimax (Z)}{\phiquantum(Z)}
    = \frac{s_1 + s_2}{s_1 + s_2} = 1
\end{equation}
holds for all $Z \in \MmR$ with $\phiquantum(Z) > 0$.
\end{proof}

\subsection{Proof of Theorem~\ref{thm:maximizer_gammasep}}
\label{app:proof_maximizer_gammasep}
\begin{proof}
Since $C \in \Cquantum$, the convex hull representation (Theorem~\ref{thm:convex_hull_characterization_of_Cquantum}) gives $C = AQB^\top$ with $Q \in \op{Conv}(\mrm{SO}^{-}(3))$.
In particular, $\ran C \subset \ran A$ and $\ran C^\top \subset \ran B$, so $\gammasep(C)$ is finite.
Since $\gammasep$ is a convex function, it attains its maximum over the convex set $\Cquantum$ at an extreme point $C = AQB^\top$ with $Q \in \mrm{SO}^{-}(3)$.
For such $C$, Theorem~\ref{thm:gauge_function_of_Csep} gives
\begin{equation}
    \gammasep(C) = \trnorm{A^{+} C\, (B^\top)^{+}} = \trnorm{A^{+}A\, Q\, B^\top (B^\top)^{+}} = \trnorm{P_{\ran A^\top}\, Q\, P_{\ran B^\top}},
\end{equation}
where $P_M$ denotes the orthogonal projection onto a subspace~$M$.

\noindent\textbf{Case (i): $r = 3$.}
Since $\rank A = \rank B = 3$, we have $P_{\ran A^\top} = P_{\ran B^\top} = I_3$, so
\begin{equation}
    \gammasep(C) = \trnorm{Q} = 3,
\end{equation}
because every $Q \in \mrm{SO}^{-}(3)$ has singular values $1, 1, 1$.
Therefore every $Q \in \mrm{SO}^{-}(3)$ achieves the maximum.

\noindent\textbf{Case (ii): $r = 2$.}
Without loss of generality, suppose $\rank A = 2$ and $\rank B = 3$ (the case $\rank A = 3$, $\rank B = 2$ is analogous).
Then $P_{\ran B^\top} = I_3$ and
\begin{equation}
    \gammasep(C) = \trnorm{P_{\ran A^\top}\, Q}.
\end{equation}
Since $(P_{\ran A^\top}\, Q)(P_{\ran A^\top}\, Q)^\top = P_{\ran A^\top}\, Q Q^\top\, P_{\ran A^\top} = P_{\ran A^\top}$, the matrix $P_{\ran A^\top}\, Q$ has singular values $1, 1, 0$.
Therefore $\gammasep(C) = 2$ for every $Q \in \mrm{SO}^{-}(3)$.

When $\rank A = \rank B = 2$, we have
\begin{equation}
    \gammasep(C) = \trnorm{P_{\ran A^\top}\, Q\, P_{\ran B^\top}}.
\end{equation}
The matrix $P_{\ran A^\top}\, Q\, P_{\ran B^\top}$ has rank at most $2$.
By the submultiplicativity of singular values under projection, its singular values $\sigma_1 \geq \sigma_2 \geq 0$ satisfy $\sigma_k \leq 1$, so $\gammasep(C) = \sigma_1 + \sigma_2 \leq 2$.
Equality $\gammasep(C) = 2$ holds if and only if $\sigma_1 = \sigma_2 = 1$, which requires $Q(\ran B^\top) = \ran A^\top$.
\end{proof}

\subsection{Proof of Theorem~\ref{thm:maximizer_gammaquantum}}
\label{app:proof_maximizer_gammaquantum}
\begin{proof}
Since $C \in \Cmax$, the convex hull representation (Theorem~\ref{thm:convex_hull_characterzation_of_Cmax}) gives $C = AQB^\top$ with $Q \in \op{Conv}(\mrm{O}(3))$.
In particular, $\ran C \subset \ran A$ and $\ran C^\top \subset \ran B$, so $\gammaquantum(C)$ is finite.
Since $\gammaquantum$ is a convex function, it attains its maximum over the convex set $\Cmax$ at an extreme point $C = AQB^\top$ with $Q \in \mrm{O}(3)$.

\noindent\textbf{Case (i): $r = 3$.}
Since $\rank A = \rank B = 3$, we have $P_{\ran A^\top} = P_{\ran B^\top} = I_3$, so
\begin{equation}
    \gammaquantum(C) = \norm{A^{+} C\, (B^\top)^{+}}_{+} = \norm{Q}_{+}.
\end{equation}
Every $Q \in \mrm{O}(3)$ has singular values $1, 1, 1$.
By definition, $\|Q\|_{+} = s_1 + s_2 + s_3\, \op{sgn}(\det Q)$.
When $\det Q = +1$, i.e., $Q \in \mrm{SO}(3)$, we obtain $\|Q\|_{+} = 1 + 1 + 1 = 3$.
When $\det Q = -1$, i.e., $Q \in \mrm{SO}^{-}(3)$, we obtain $\|Q\|_{+} = 1 + 1 - 1 = 1$.
Therefore the maximum $\gammaquantum(C) = 3$ is achieved precisely when $Q \in \mrm{SO}(3)$.

\noindent\textbf{Case (ii): $r \leq 2$.}
By Theorem~\ref{thm:gauge_function_of_Cquantum}, $\gammaquantum(C) = \gammamax(C) = \opnorm{P_{\ran A^\top}\, Q\, P_{\ran B^\top}}$ when $r \leq 2$.
Since orthogonal projections are contractive in the operator norm, $\opnorm{P_{\ran A^\top}\, Q\, P_{\ran B^\top}} \leq \opnorm{Q} = 1$ for every $Q \in \mrm{O}(3)$.
Moreover, $\Cmax = \Cquantum$ when $r \leq 2$, so $\Cmax$ is the unit ball of $\gammaquantum$ and hence $\sup_{C \in \Cmax} \gammaquantum(C) = 1$.
The maximizers are precisely the points with $\gammaquantum(C) = 1$, i.e., $\partial \Cmax = \partial \Cquantum$.

For an extreme point $C = AQB^\top$, the condition $\gammaquantum(C) = 1$ reads $\opnorm{P_{\ran A^\top}\, Q\, P_{\ran B^\top}} = 1$, which we now show is equivalent to $\ran A^\top \cap Q(\ran B^\top) \neq \{0\}$.
If $\opnorm{P_{\ran A^\top}\, Q\, P_{\ran B^\top}} = 1$, there exists a unit vector $\bs{v} \in \ran B^\top$ with $\|P_{\ran A^\top}\, Q\, \bs{v}\| = 1$; since $\|Q\,\bs{v}\| = 1$, this forces $Q\,\bs{v} \in \ran A^\top$, hence $Q\,\bs{v} \in \ran A^\top \cap Q(\ran B^\top)$.
Conversely, if $0 \neq \bs{u} \in \ran A^\top \cap Q(\ran B^\top)$, write $\bs{u} = Q\,\bs{v}$ with $\bs{v} \in \ran B^\top$; normalizing gives $P_{\ran B^\top}\,\hat{\bs{v}} = \hat{\bs{v}}$ and $P_{\ran A^\top}\,Q\,\hat{\bs{v}} = Q\,\hat{\bs{v}}$, so $\opnorm{P_{\ran A^\top}\, Q\, P_{\ran B^\top}} = 1$.

Finally, when $\rank A \geq 2$ and $\rank B \geq 2$, we have $\dim(\ran A^\top) + \dim(Q(\ran B^\top)) \geq 2 + 2 = 4 > 3 = \dim \mbb{R}^3$, so $\ran A^\top \cap Q(\ran B^\top) \neq \{0\}$ for every $Q \in \mrm{O}(3)$, and every extreme point is a maximizer.
\end{proof}

\section{Proof of Proposition~\ref{prop:rigidity}}\label{app:rigidity}

\begin{proof}
The block-positivity condition $\langle x \otimes y | \hat{\rho} | x \otimes y \rangle \geq 0$ for all product vectors, expressed in terms of the Bloch-sphere parametrization, reads
\begin{equation}\label{eq:bp_general}
    1 + \bs{r}_A \cdot \bs{u} + \bs{r}_B \cdot \bs{v} + \bs{u}^\top T \bs{v} \geq 0 \qquad (\forall\, \bs{u}, \bs{v} \in \mbb{R}^3,\; \|\bs{u}\| = \|\bs{v}\| = 1).
\end{equation}

\noindent\textbf{Step 1: $\bs{r}_A = T \bs{r}_B$.}
Setting $\bs{v} = -T^\top \bs{u}$ (which satisfies $\|\bs{v}\| = 1$ since $T \in \mrm{O}(3)$) gives $\bs{u}^\top T \bs{v} = -\|\bs{u}\|^2 = -1$, so~\eqref{eq:bp_general} reduces to
\begin{equation}
    (\bs{r}_A - T \bs{r}_B) \cdot \bs{u} \geq 0 \qquad (\forall\, \|\bs{u}\| = 1),
\end{equation}
which forces $\bs{r}_A = T \bs{r}_B$.

\noindent\textbf{Step 2: Reduction to $T = I_3$.}
Substituting $\bs{w} := T \bs{v}$ (so $\|\bs{w}\| = 1$) and using $\bs{r}_B \cdot \bs{v} = (T \bs{r}_B) \cdot \bs{w} = \bs{r}_A \cdot \bs{w}$, the condition~\eqref{eq:bp_general} becomes
\begin{equation}\label{eq:bp_identity_case}
    1 + \bs{r}_A \cdot \bs{u} + \bs{r}_A \cdot \bs{w} + \bs{u} \cdot \bs{w} \geq 0 \qquad (\forall\, \|\bs{u}\| = \|\bs{w}\| = 1).
\end{equation}

\noindent\textbf{Step 3: $\bs{r}_A = \bs{0}$.}
For fixed $\bs{u}$, minimizing over $\bs{w}$ gives $\min_{\|\bs{w}\|=1} (\bs{r}_A + \bs{u}) \cdot \bs{w} = -\|\bs{r}_A + \bs{u}\|$, so a necessary condition is
\begin{equation}
    1 + \bs{r}_A \cdot \bs{u} - \|\bs{r}_A + \bs{u}\| \geq 0 \qquad (\forall\, \|\bs{u}\| = 1).
\end{equation}
If $\bs{r}_A \neq \bs{0}$, write $\bs{r}_A = t\, \bs{e}$ with $t > 0$ and $\|\bs{e}\| = 1$, and set $s = \bs{e} \cdot \bs{u} \in [-1, 1]$.
The inequality becomes $1 + ts \geq \sqrt{1 + t^2 + 2ts}$.
Squaring both sides (valid since the left side must be nonneg.\ by assumption) yields $t^2 s^2 \geq t^2$, i.e., $s^2 \geq 1$ for all $s \in [-1, 1]$, a contradiction.
Therefore $\bs{r}_A = \bs{0}$, and $\bs{r}_B = T^\top \bs{r}_A = \bs{0}$.
\end{proof}

\section{Intrinsic determinant identity}
\label{app:det_identity}

\begin{lemma}\label{lem:det_intrinsic}
Let $A,B \in \mbb{R}^{M \times 3}$ with columns $\bs{a}_1,\ldots,\bs{a}_M$ and $\bs{b}_1,\ldots,\bs{b}_M$ viewed as row vectors, and let $Z = (z_{il}) \in \mbb{R}^{M \times M}$.
Define the cubic tensors $T^A_{ijk} := \det[\bs{a}_i,\bs{a}_j,\bs{a}_k]$ and $T^B_{lmn} := \det[\bs{b}_l,\bs{b}_m,\bs{b}_n]$.
Then
\begin{equation}
    \det(A^\top Z B)
    = \frac{1}{6}\sum_{i,j,k=1}^{M}\sum_{l,m,n=1}^{M}
      T^A_{ijk}\, z_{il}\, z_{jm}\, z_{kn}\, T^B_{lmn}.
    \label{eq:det_intrinsic_appendix}
\end{equation}
\end{lemma}

\begin{proof}
For a $3\times 3$ matrix $X$, the determinant admits the Levi-Civita expansion
\begin{equation}
    \det X = \frac{1}{6}\sum_{\alpha,\beta,\gamma=1}^{3}\sum_{\mu,\nu,\lambda=1}^{3}
    \varepsilon_{\alpha\beta\gamma}\,\varepsilon_{\mu\nu\lambda}\,
    X_{\alpha\mu}\,X_{\beta\nu}\,X_{\gamma\lambda},
    \label{eq:levi_civita_det}
\end{equation}
where $\varepsilon_{\alpha\beta\gamma}$ is the Levi-Civita symbol.
We apply this to $X = A^\top Z B$, whose entries are
\begin{equation}
    (A^\top Z B)_{\alpha\mu}
    = \sum_{i=1}^{M} \sum_{l=1}^{M} A_{i\alpha}\, z_{il}\, B_{l\mu}.
\end{equation}
Substituting into~\eqref{eq:levi_civita_det} and rearranging:
\begin{align}
    \det(A^\top Z B)
    &= \frac{1}{6}\sum_{\alpha\beta\gamma}\sum_{\mu\nu\lambda}
       \varepsilon_{\alpha\beta\gamma}\,\varepsilon_{\mu\nu\lambda}
       \left(\sum_{i,l} A_{i\alpha}\, z_{il}\, B_{l\mu}\right)
       \left(\sum_{j,m} A_{j\beta}\, z_{jm}\, B_{m\nu}\right)
       \left(\sum_{k,n} A_{k\gamma}\, z_{kn}\, B_{n\lambda}\right) \notag\\[4pt]
    &= \frac{1}{6}\sum_{i,j,k}\sum_{l,m,n}
       z_{il}\, z_{jm}\, z_{kn}
       \underbrace{\left(\sum_{\alpha\beta\gamma}
       \varepsilon_{\alpha\beta\gamma}\, A_{i\alpha}\, A_{j\beta}\, A_{k\gamma}\right)}_{=\,T^A_{ijk}}
       \underbrace{\left(\sum_{\mu\nu\lambda}
       \varepsilon_{\mu\nu\lambda}\, B_{l\mu}\, B_{m\nu}\, B_{n\lambda}\right)}_{=\,T^B_{lmn}},
\end{align}
where we used $\sum_{\alpha\beta\gamma} \varepsilon_{\alpha\beta\gamma}\, A_{i\alpha}\, A_{j\beta}\, A_{k\gamma} = \det[\bs{a}_i,\bs{a}_j,\bs{a}_k] = T^A_{ijk}$, and similarly for~$T^B_{lmn}$.
\end{proof}

\bibliography{references}  

\end{document}